\documentclass[11pt]{article}
\usepackage[margin=1in]{geometry}
\usepackage{amsthm,amsfonts,amsmath,amssymb}
\usepackage{enumerate,graphicx}
\usepackage{xspace}
\usepackage{boxedminipage}
\usepackage{url}
\usepackage{accents}
\usepackage{multirow}
\usepackage{booktabs}
\usepackage{enumitem}
\usepackage{wrapfig}

\usepackage{hyperref}
\hypersetup{%
   breaklinks,%
   colorlinks=true,%
   linkcolor=black,%
   urlcolor=black,%
   citecolor=[rgb]{0,0,0.45}
}

\usepackage{color}
\usepackage{graphicx}
\usepackage{natbib}
\usepackage{tikz}

\usepackage[utf8]{inputenc} 

\newtheorem*{theorem*}{Theorem}

\newtheorem{lemma}{Lemma}[section]
\newtheorem{theorem}[lemma]{Theorem}

\newtheorem{corollary}[lemma]{Corollary}

\newtheorem{claim}[lemma]{Claim}
\newtheorem{proposition}[lemma]{Proposition}

\theoremstyle{definition}
\newtheorem{remark}[lemma]{Remark}

\newtheorem{definition}[lemma]{Definition}

\def\ceil#1{\lceil {#1} \rceil}

\def\poly{\mathrm{poly} }
\def\polylog{\mathrm{ polylog} }

\def\sol{\mathrm{SOL}}
\def\opt{\mathrm{OPT}}

\def\auglp{\textsc{AugLP}}

\usepackage{algorithm}
\usepackage{algpseudocode}
\usepackage{mathtools}

\def\eps{\epsilon}
\def\even{\mathrm{even}}
\def\odd{\mathrm{odd}}

\definecolor{gray}{rgb}{0.5,0.5,0.5}

\title{Streaming Algorithms for Connectivity Augmentation}
\author{Ce Jin\thanks{MIT. Work done during an internship at Microsoft Research--Redmond. Email:~\href{mailto:cejin@mit.edu}{cejin@mit.edu}} \and Michael Kapralov\thanks{EPFL. Email:~\href{mailto:michael.kapralov@epfl.ch}{michael.kapralov@epfl.ch}} \and Sepideh Mahabadi\thanks{Microsoft Research--Redmond. Email:~\href{mailto:smahabadi@microsoft.com}{smahabadi@microsoft.com}} \and Ali Vakilian\thanks{Technological Institute at Chicago (TTIC). Email:~\href{mailto:vakilian@ttic.edu}{vakilian@ttic.edu}}}
\date{}

\begin{document}
\maketitle
\begin{abstract}
    We study the $k$-connectivity augmentation problem ($k$-CAP) in the single-pass streaming model. Given a $(k-1)$-edge connected graph $G=(V,E)$ that is stored in memory, and a stream of weighted edges (also called links) $L$ with weights in $\{0,1,\dots,W\}$, the goal is to choose a minimum weight subset $L'\subseteq L$ of the links such that $G'=(V,E\cup L')$ is $k$-edge connected. We give a $(2+\epsilon)$-approximation algorithm for this problem which requires to store $O(\epsilon^{-1} n\log n)$ words. Moreover, we show the tightness of our result: Any algorithm with better than $2$-approximation for the problem requires $\Omega(n^2)$ bits of space even when $k=2$. This establishes a gap between the optimal approximation factor one can obtain in the streaming vs the offline setting for $k$-CAP.
    
    We further consider a natural generalization to the fully streaming model where both $E$ and $L$ arrive in the stream in an arbitrary order. We show that this problem has a space lower bound that matches the best possible size of a spanner of the same approximation ratio. Following this, we give improved results for spanners on weighted graphs: We show a streaming algorithm that finds a $(2t-1+\epsilon)$-approximate weighted spanner of size at most $O(\epsilon^{-1} n^{1+1/t}\log n)$ for integer $t$, whereas the best prior streaming algorithm for spanner on weighted graphs had size depending on $\log W$. We believe that this result is of independent interest. Using our spanner result, we provide an optimal $O(t)$-approximation for $k$-CAP in the fully streaming model with $O(nk + n^{1+1/t})$ words of space. 

    Finally we apply our results to network design problems such as Steiner tree augmentation problem (STAP), $k$-edge connected spanning subgraph ($k$-ECSS) and the general Survivable Network Design problem (SNDP). In particular, we show a single-pass $O(t\log k)$-approximation for SNDP using $O(kn^{1+1/t})$ words of space, where $k$ is the maximum connectivity requirement.
\end{abstract}

\section{Introduction}
In the (weighted) {\em $k$-connectivity augmentation} problem ($k$-CAP), given a $(k-1)$-edge-connected $n$-vertex graph $G = (V, E)$ (possibly with parallel edges) together with a set of weighted candidate edges (also called {\em links}) denoted by $L \subseteq \binom{V}{2}$ and their weights $w\colon L \rightarrow \{0,1,\dots,W\}$, the goal is to find a minimum weight subset $S \subseteq L$ of the links such that $(V, E\cup S)$ is $k$-edge-connected.
Augmenting connectivity is a crucial task for enhancing network reliability which can be used for strengthening the resilience of a network and ensuring uninterrupted access for all users.
$k$-CAP is among the most elementary questions in Network Design, which is an important area of discrete optimization. 
The iterative rounding method of~\citet{Jain01} provides a $2$-approximation for a more general problem of {\em survivable network design} problem (SNDP). Untill very recently, nothing better than $2$ approximation was known even for weighted {\em tree augmentation problem} (TAP). In a recent development, weighted $k$-CAP has witnessed breakthroughs with approximation factors below $2$~\citep{traub2022better,traub2022local,traub20231}. The state-of-the-art for weighted $k$-CAP is $1.5+\eps$ approximation.

In this work, we consider weighted $k$-CAP in the streaming model, which is one of the most common models for processing real-time and large-scale data.  
A graph streaming algorithm operates by processing a sequence of graph edges presented in any order (or in some applications in random order), reading them one by one.
The primary objective is to design algorithms that can process the entire edge sequence and output an approximately efficient solution, making just one (or a few passes), while utilizing limited memory resources. Ideally, the space usage of the algorithm should be significantly smaller than the size of the $n$-vertex input graph (with possibly $O(n^2)$ edges), preferably $O(n\cdot \mathrm{polylog}(n))$ memory, which is referred to as the \emph{semi-streaming} model~\citep{feigenbaum2004graph}.

While graph problems such as minimum spanning tree~\citep{ahn2012graph,sun2015tight,nelson2019optimal}, matching~\citep{DBLP:conf/approx/McGregor05,GoelKK12,assadi2016maximum,assadi2017estimating,DBLP:conf/soda/Kapralov21}, spanners, sparsifiers and shortest paths~\citep{FeigenbaumKMSZ08,DBLP:journals/ipl/Baswana08,DBLP:journals/talg/Elkin11,ahn2012graph,kapralov2014spanners,guruswami2016superlinear,fernandez2020graph,filtser2021graph} have received significant attention in the streaming model, the connectivity augmentation problem, has received comparatively very limited study in this context. Prior to our result, only testing $k$-connectivity in streaming was studied~\citep{zelke2006k,DBLP:conf/esa/CrouchMS13,sun2015tight}, which showed that testing $k$-edge-connectivity in streaming requires $\tilde \Theta(nk)$ space in one pass, and $\tilde \Theta(n)$ space in two passes \citep{DBLP:conf/innovations/RubinsteinSW18,assadi2021simple}. 

\subsection{Our Computational Models}\label{sec:model}
In this work, we study graph augmentation problems in the streaming model of computation. The input to the $k$-CAP problem consists of two pieces of information, namely the $(k-1)$-connected network $G$ and the set of links that can be used to augment connectivity. 

\paragraph{Link arrival streaming.} In the link arrival streaming model the graph $G$ is presented to the algorithm first, and the cost of storing it does not count towards the space complexity of the algorithm. This is akin to the oracle model that is routinely used to study submodular function maximization in the streaming model (e.g., in~\citep{badanidiyuru2014streaming,norouzi2018beyond}): One thinks of having an oracle for the function being maximized. For submodular function maximization it is not always clear how to implement this oracle in small space, but in our case the actual cost of storing a sufficient representation of the graph $G$ can be easily made $O(nk)$, and, with some work, even $O(n)$, as we now explain. 

Note that a minimally $k$-connected graph has size $O(nk)$. So if the graph has larger size, one can process the edges of $G$ (even in a streaming fashion) using a $k$-connectivity certificate of $G$ that preserves all cuts of value at most $k$, and store this compact representation in $O(nk)$ space. Finally, one can apply even a more efficient preprocessing that preserves a similar information via a {\em cactus} graph with $O(n)$ edges. Then the problem becomes streaming {\em cactus augmentation}. The cactus augmentation problem itself is a well-studied problem in particular for designing approximation algorithms for $k$-CAP.
To simplify the notation, throughout the paper, we assume the latter compact representation of size $O(n)$.

\paragraph{Fully streaming.} Besides the most natural link arrival model defined above, we study the more general model where the edges of $G$ and the links that can be used for augmentation may arrive in an interleaved fashion.  This model is quite general: in particular, it allows for the edges of $G$ to arrive {\em after} the links, in which case the algorithm must maintain a compressed representation of the stream of links that allows augmenting any given graph $G$ presented later!

For the other graph problems studied in this paper, namely spanner, SNDP and $k$-edge connected spanning subgraph ($k$-ECSS), we consider the standard edge arrival streams in which edges of the input graph arrives one by one in an arbitrary order stream.

\subsection{Our Results}
In this paper, we focus on insertion-only streams, and provide the first streaming algorithms for $k$-CAP in link arrival streams and fully streaming. Table~\ref{tab:results} summarizes our results.

\paragraph{Graph augmentation in link arrival.} 
We show tight results for weighted $k$-CAP in link arrival streams (see first row in Table~\ref{tab:results}). Note that, while we can achieve a factor $2+\epsilon$ approximation in $O(\frac{n}{\eps}\log n)$ words of space, our lower bound shows that getting better than $2$ approximation requires $\Omega(n^2)$ bits of memory. This establishes a gap between the streaming setting and the offline setting where strictly better than $2$ approximation algorithms are known (e.g., see~\citep{traub20231}). 
An easy argument shows that $\Omega(n)$ bits of space is necessary for achieving any approximation for $k$-CAP  in link arrival streams (Proposition~\ref{prop:lb-trivial-link-arrival}), so our algorithm has nearly-tight space complexity.
If one picks a $k$-connectivity certificate as the compact representation of $G$, the space complexity of the upper bound becomes $O(nk +\frac{n}{\eps}\log n)$.

Further, we study the Steiner tree augmentation problem (STAP) which is a generalization of the tree augmentation problem (TAP) in link arrival streams and provide matching upper and lower bounds (See the second row in Table~\ref{tab:results}). While our lower bound holds for link arrival streams, our algorithm works even in the more general fully streaming too. We remark that, while in the offline setting TAP and STAP admit similar approximations~\citep{ravi2022new}, there is a gap in their complexities in the streaming model.

\begin{table}[!ht]
\centering
{\renewcommand{\arraystretch}{1.2}%
\begin{tabular}{c|c|c|c|c|c}\toprule
{\bf Problem} & {\bf Pass} & {\bf Approx.} & {\bf Space} & {\bf Stream} & {\bf Notes}\\ \midrule\midrule
\multirow{4}{*}{$k$-CAP} & \multirow{4}{*}{$1$} & $2+\eps$ & $O(\frac{n}{\eps}\log n)$ & \multirow{2}{*}{link arrival} & Theorem~\ref{thm:link-arrival-alg}\\
& & $2-\eps$ & $\Omega(n^2)$ bits & & Theorem~\ref{thm:lb-link-arrival}\\
\cmidrule{3-5}
& & \multirow{2}{*}{$O(t)$} & $\tilde{O}(kn + n^{1+\frac{1}{t}})$  & \multirow{2}{*}{fully streaming}& Theorem~\ref{thm:fullystreamingkconn}\\
& & & $\Omega(kn + n^{1+\frac{1}{t}})$ bits & & Theorem~\ref{thm:lowerbound-combined}\\ 
\midrule
\multirow{2}{*}{STAP} & \multirow{2}{*}{$1$} & \multirow{2}{*}{$O(t)$} & $\tilde{O}(n^{1+\frac{1}{t}})$ & fully streaming & Corollary~\ref{cor:steinerfully}\\ 
&  &  & $\Omega(n^{1+\frac{1}{t}})$ bits & link arrival & Corollary~\ref{cor:lowerbound-steiner}\\
\midrule\midrule
\multirow{2}{*}{Spanner} & \multirow{2}{*}{$1$} & \multirow{2}{*}{$O(t)$} & $\tilde{O}(n^{1+\frac{1}{t}})$ & \multirow{2}{*}{edge arrival} & Theorem~\ref{thm:spanner-alg} \\
& & & $\Omega(n^{1+\frac{1}{t}})$ bits & & \footnotesize{Erd\H{o}s' girth conjecture}\\ \midrule\midrule
\multirow{2}{*}{SNDP} & \multirow{2}{*}{1} & $O(t \log k)$ & $\tilde{O}(kn^{1+\frac{1}{t}})$ & \multirow{2}{*}{edge arrival} &  Theorem~\ref{thm:sndp} \\ 
& & $O(t)$ & $\Omega(n^{1+\frac{1}{t}})$ bits & &Corollary~\ref{cor:lowerbound-steiner}
 \\ \midrule
$k$-ECSS & $k$ & $O(\log k)$ & $O(kn\log n)$ & edge arrival & Corollary~\ref{thm:k-ecss-alg}\\ 
\bottomrule
\end{tabular}
}
\caption{Summary of our results for $k$-CAP, STAP, Spanner and SNDP in steaming models.
All our problems are \emph{weighted}.
The space upper bounds are measured in words, while the lower bounds are in bits. We use $\tilde O(f)$ to mean $O(f\cdot \mathrm{polylog} f)$ (it does not hide $\log W$ factors).
All our algorithms are deterministic, whereas all lower bounds hold for randomized algorithms with constant success probability.
}
\label{tab:results}
\end{table}

\paragraph{Graph augmentation in fully streaming.}
We further show matching upper and lower bounds (up to a $\polylog(n)$ factor) for $k$-CAP in the fully streaming setting (see the lower section in the first row of Table~\ref{tab:results}).
The main component in our algorithm for solving $k$-CAP is an improved streaming algorithm for constructing \emph{spanners on weighted graphs}. 
In particular, our upperbound implies that spanner is an optimal ``universal'' augmentation set for $k$-CAP. 

\paragraph{Improved streaming spanner in weighted graphs.} Given an $n$-vertex graph $G = (V, E)$ with a weight function $w\colon E \rightarrow \{0,\dots, W\}$, a subgraph $H \subseteq G$ is a $t$-spanner of $G$ if for every $(u,v)\in E$, the shortest $uv$-path in $H$ has weight at most $t \cdot w(uv)$. In streaming spanner, which is a well-studied problem~
\citep{DBLP:journals/ipl/Baswana08,DBLP:journals/talg/Elkin11,ahn2012graph,kapralov2014spanners,fernandez2020graph}, edges of $E$ arrive in an arbitrary order stream.
While by using the standard weight-based partitioning trick, constructing an $O(t)$-spanner in $O(n^{1 + 1/t} \cdot \log W)$ words of space in one pass over the stream is straightforward (e.g., mentioned in~\citep{filtser2021graph}), it was not known whether the dependence on $\log W$ is crucial.\footnote{We remark that our contribution in removing the dependence on $\log W$ from the number of edges in spanner (and consequently from 
$k$-CAP) is conceptually interesting, as most graph streaming algorithms are mainly designed for unweighted graphs, and extending them to the weighted case typically incurs a $\log W$ loss.} 
Exploiting an even-odd bucketing approach, we provide a streaming algorithm with space complexity $O(n^{1 + 1/t} \cdot \log \min(W,n))$ words which by the well-known Erd\H{o}s girth conjecture is basically the best one can hope for up to logarithmic factors. 
We further apply this even-odd bucketing to the $k$-CAP problem in the link arrival setting, and obtain a (more technical) algorithm (Theorem~\ref{thm:link-arrival-alg}) with no dependence on $\log W$ in its space complexity.

\paragraph{Streaming SNDP.} Finally, we describe an application of our results for designing the {\em first} one-pass streaming algorithms for the problem in insertion only edge arrival streams, where the edges of the input graph arrive in an arbitrary order stream. 

In SNDP, given a graph $G=(V, E)$ with a weight function $w: E \rightarrow \{0, 1, \dots, W\}$ together with a \emph{connectivity requirement} $r \colon V\times V \rightarrow \mathbb{Z}_{\ge 0}$, the goal is to find a minimum weight subgraph $H\subseteq G$ so that for every $s,t\in V$, $H$ contains $r(st)$ edge-disjoint paths connecting $s$ and $t$. A parameter of interest in SNDP is the maximum connectivity requirement $k = \max_{st} r(st)$.
SNDP is a classic problem in combinatorial optimization and generalizes several well-studied problems such as MST, Steiner tree, $k$-edge connected spanning subgraph ($k$-ECSS), and $k$-CAP. 

The fourth row of Table \ref{tab:results} shows our results for SNDP in edge arrival streams. In fact, our streaming algorithm works even for the more general problem of covering proper functions of the form $f: 2^V \rightarrow \{0,1, \dots, k\}$ using the edges of $G$ (see Section~\ref{sec:sndp} for more details).  

$k$-ECSS, which itself is a basic problem in discrete optimization, is a variant of SNDP in which for every $s,t \in V$, $r(st) = k$. As a straightforward application of our algorithm for $k$-CAP in link arrival streams, we get a $k$-pass, $O(\log k)$-approximation for $k$-ECSS using $O(kn\log n)$ words of space. (See last row of Table \ref{tab:results}). 

\paragraph{Unweighted variant.} 
We remark that while we get tight algorithms for weighted $k$-CAP in both link arrival and fully streaming models, our lower bounds for link arrival does not hold for unweighted graphs.
In Corollary \ref{col:lb-unweighted}, by a reduction from bipartite matching and invoking the result of \cite{DBLP:conf/soda/Kapralov21}, we observe a weaker lower bound that no streaming algorithm with $n\polylog n$ space can achieve an approximation factor better than $1.409$.
Therefore, it remains an interesting open question to close the gap between $1.409$ and $2$ for unweighted $k$-CAP.
Again, given that this lower bound is for tree augmentation, and the best known algorithm for (offline) TAP in unweighted graphs achieves an approximation factor of $1.326$ \citep{garg2023improved}, this again shows a gap between the two models for the problem in the unweighted variant.

\subsection{Our Techniques}
Given a streaming algorithm for the unweighted variants of both $k$-CAP and the spanner problems, an easy generalization to the weighted graphs is by partitioning the set of weights into $\log_{1+\eps} W$ number of classes and roughly running the unweighted sparsification on each class, resulting in $\eps^{-1}\log W$ blow up in the space usage. To remove the dependency on $\log W$ from the number of words, we follow an \emph{even-odd bucketing} approach. More precisely, we partition the weights into much larger classes (i.e., \emph{buckets}), such that the minimum and maximum weight in each class differ by $\poly(n)/\epsilon$. This ensures that first, inside each class one can perform the weight-based partitioning to solve the problem while having only $\log n$ dependence in the space. Second, even picking all the edges from the $(i-2)$-th class $E_{i-2}$ is cheaper than picking any edge in the $i$-th class $E_i$ (i.e., it only introduces an extra $(1+\eps)$ multiplicative factor). 
This assumption allows us to infer additional properties about the graph once we are processing the edges in the class $E_i$, and shrink the problem significantly from each level $E_{i-2}$ to $E_i$. Thus our algorithm proceeds by separating the sparsification for the even-indexed buckets $E_{2i}$ and the odd-indexed buckets $E_{2i-1}$, and processes the buckets from smallest to largest weights.

\paragraph{Spanner.} First, consider the spanner problem, and let $\mathcal{C}=\{C_1,\dots,C_r\}$ be the set of connected components created by the edges from the classes upto $E_{i-2}$. The even-odd bucketing ensures that we only need to consider the edges from $E_i$ that are between two different components of $\mathcal{C}$. 
Thus, we shrink each connected component into a super-node and use the standard spanner algorithm with weight-based partitioning on this reduced graph.
Note that the space usage of the algorithm is proportional to the number of \emph{super-nodes with non-zero degree}. However, all such super-nodes will merge into bigger components for the next bucket $E_{i+2}$. Therefore the space usage of the algorithm for processing $E_i$ can be charged to the reduction in the number of super-nodes. 
Since the number of super-nodes starts from $n$ and goes down to $1$, the total space usage of the algorithm can be bounded as a function of $n$.
Finally, we need to perform the above process in a streaming setting: As we receive more edges in the stream, the components in $\mathcal{C}_i$ change but it is easy to maintain all required information in a streaming fashion.

\paragraph{Link arrival $k$-CAP.}
Our algorithm for $k$-CAP is more involved. First, by standard results in the literature, the problem reduces to cycle augmentation: given a cycle $C$, the goal is to augment it with a subset of edges from $L$ such that the resulting graph becomes $3$-edge-connected. Let the nodes on the cycle be indexed $1$ to $n$ in this order with vertex $1$ being called the \emph{root}. Now every cut of size $2$ corresponds to two edges on the cycle. We specify such a cut with the interval $[i,j]$ with $1 < i\leq j\leq n$ that does not include the root. The goal is to cover all such cuts specified by these intervals.

First, using known ideas from~\citep{KhullerT93,KhullerV94}, we present a simple streaming algorithm for the unweighted variant of the problem as follows (See Figure~\ref{fig:cycle-aug}). We replace every link $uv$ by two directed links $\vec{uv}$ and $\vec{vu}$, (this is where the factor $2$ in the approximation comes from), and we say that $\vec{uv}$ covers a cut $[i,j]$ if $v\in [i,j]$ and $u\notin [i,j]$. Now one can show that for $1\leq u < u' < v$, it is always better to keep the edge $\vec{uv}$ than $\vec{u'v}$. Similarly, for $v < u' < u \leq n$, it is always better to keep the edge $\vec{uv}$ than $\vec{u'v}$. As a result, for each vertex, we keep at most two incoming edges. Therefore, the total space usage of the algorithm is only $O(n)$ in this case. Again this algorithm can be generalized to the weighted graphs using a weight-based partitioning, introducing a factor $\log W$.
\begin{figure}
\centering
\includegraphics[width=.4\textwidth]{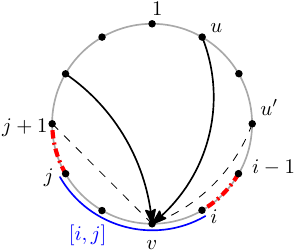}
\caption{This figure illustrates the process of edge sparsification in a cycle for connectivity augmentation.}\label{fig:cycle-aug}
\end{figure}

To remove the dependency on $\log W$, again we consider the even-odd bucketing. This time, for each weight class $E_i$, we consider the $3$-edge-connected components $C_1,\dots, C_r$ formed by the edges in buckets up to $E_{i-2}$. Again using the even-odd bucketing plus the fact that the cycle is already $2$-edge-connected, we can show that shrinking each of the $3$-connected components into a super-node still works. The main challenge is that as opposed to the spanner setting, the problem on the super-node does not reduce to the same problem of cycle augmentation. This is because a single super-node does not necessarily span a consecutive set of vertices on the cycle. However, we note that in this case, the min-cuts on the cycle that do not fully include or fully exclude the vertices in a single super-node do not need to be considered. This allows us to reduce the space usage of the algorithm again to be proportional to the number of super-nodes and thus bound the total space usage of the algorithm as a function of $n$.

\paragraph{Fully streaming $k$-CAP.} Our algorithm in this setting maintains two sketches. 
First, it keeps a $k$-connectivity certificate on the set of edges $E$ using a folklore streaming algorithm that keeps $k$ disjoint forests, which contains the information of all min-cuts of $E$ that need to be augmented in $k$-CAP.
Second, employing our results on weighted spanners, the algorithm maintains a spanner for the set of (weighted) links. This means that every link $\ell$ of weight/length $w$  that we miss, can be replaced with a path of weight at most $O(t)\cdot w$, thus covering all the min-cuts originally covered by $\ell$. We show that this is a near-optimal algorithm one can get in this setting.

\paragraph{Lower bounds.} Most of our lower bounds are via simple reductions from the INDEX problem in a two-party communication model, where we embed the bit-string held by Alice into edges of a graph, where by asking augmentation queries, Bob is able to tell whether edge $(u,v)$ exists in Alice's graph for any pair of vertices $u,v$.
The most interesting one of our lower bounds (Theorem~\ref{thm:lowerbound-t}) shows that, in the fully streaming model, the space complexity for storing a spanner is essentially necessary. In the proof we let Alice hold a subgraph of a high-girth graph, and Bob wants to estimate the distance in this graph between $u,v$ (which is sufficient for telling whether $(u,v)$ is an edge, due to the high girth).
Our proof reduces this problem of estimating the distance between $u,v$ to the problem of augmenting a chain with end points $u,v$ into a $2$-edge-connected graph. However, we also need rule out potential augmentation solutions that do not correspond to a $uv$-path.

\paragraph{Applications.}
Our algorithms for streaming connectivity augmentation also imply streaming algorithms for problems such as STAP, $k$-ECSS and SNDP.
In particular, our one-pass algorithm for SNDP works by running $k$ instances of our streaming spanner algorithm in parallel, which store $k$ disjoint sparse subgraphs of the input graph that satisfy certain approximation guarantee. In particular, we show these $k$ disjoint ``spanner-like'' objects forms a coreset for SNDP instances with maximum connectivity requirement at most $k$.\footnote{In fact, the coreset guarantee holds even for the more general covering proper functions of the form $f: 2^V \rightarrow \{0,1, \dots, k\}$.} 
Our approach follows the augmentation framework of \citep{williamson1993primal,goemans1994improved} to show the existence of an approximately good solution using edges from these $k$ sparse subgraphs.

\subsection{Related Work}\label{app:related-work}
\paragraph{Approximation algorithms of $k$-CAP.} 
The edge-connectivity of a graph plays a central role in a wide range of network design problems, spanning both classical and modern problems. 
While the celebrated iterative rounding technique of~\citep{Jain01} provides a $2$-approximation for most of these problems, any better than $2$-approximation for them are among main open problems within the field of approximation algorithms. 

Significant progress has been made in achieving better than a factor of $2$-approximation for specific instances of the weighted $k$-CAP.\footnote{For this problem, $2$-approximation can be obtained via both standard primal-dual~\citep{goemans1994improved} and iterative rounding~\citep{Jain01} techniques as well as combinatorial approach of~\citep{KhullerT93}.} Notably, extensive research focusing on the well-studied unweighted TAP has led to breakthroughs~\citep{nagamochi2003approximation,even20091,kortsarz2015simplified,grandoni2018improved,cecchetto2021bridging}, culminating in an approximation factor of $1.326$~\citep{garg2023improved}. Remarkably, this same factor has also been achieved for the unweighted $k$-CAP~\citep{cecchetto2021bridging}, a problem that recently saw significant advancements surpassing the 2-approximation barrier~\citep{byrka2020breaching}. Moreover, in a recent development, the weighted TAP and $k$-CAP have witnessed breakthroughs with approximation factors below $2$~\citep{traub2022better,traub2022local,traub20231}. It is noteworthy that these advancements in the weighted variants are relatively recent in the research landscape.

The Steiner tree augmentation problem, in which given a Steiner tree $T \subset G = (V, E)$ over terminals $R\subset V$ the goal is to find a minimum weight set of edges $H\subseteq G\setminus T$ that increases the connectivity of the set $R$ to $2$, has also been studied and recently: \cite{ravi2022new} provides $(1.5+\eps)$-approximation generalizing some of the techniques in~\citep{traub2022local}. More generally, the augmentation framework is among the classical techniques for designing approximation algorithms for general connectivity problems~\citep{williamson1993primal,goemans1994improved,nutov2010approximating}.

\paragraph{SNDP.} Similarly to $k$-ECSS, the augmentation variant of SNDP has been extensively studied and is significant in the development of approximation algorithms for different variations of SNDP.
Notably, the augmentation variant of SNDP generalizes well-studied problems such as TAP, STAP and $k$-CAP. 
The augmentation variant of SNDP was originally studied to analyze the primal-dual methods for SNDP, leading to $k$ and $\log k$ approximations~\citep{williamson1993primal,goemans1994improved}, and compared to the state-of-the-art $2$-approximation iterative rounding technique of~\citep{Jain01} has the advantage of applicability to other variants of SNDP such as node-weighted SNDP~\citep{nutov2010approximating,chekuri2021node,chekuri2012prize} or vertex-connectivity SNDP~\citep{kortsarz2005approximating,fakcharoenphol2008log2,cheriyan2014approximating}.

\paragraph{Spanners and sparsifiers.}
Graph spanners are important tools for graph compression in which the distances between the nodes are preserved. See \citep{DBLP:journals/csr/AhmedBSHJKS20} for a survey on graph spanners in general. Spanners have also been studied extensively in the streaming setting, see e.g., \citep{DBLP:journals/ipl/Baswana08,ahn2012graph,DBLP:journals/talg/Elkin11,kapralov2014spanners,filtser2021graph}. For other notions of graph sparsifiers in the streaming model, see e.g., \citep{DBLP:journals/siamcomp/KapralovLMMS17,DBLP:conf/soda/KapralovMMMNST20}.

\subsection{Organization.}
In Section~\ref{sec:linkarrive}, we present our $(2+\eps)$-approximate algorithms for $k$-CAP in the link arrival model, and present a lower bound showing our approximation ratio is close to optimal.
In Section~\ref{sec:fullystream}, we study $k$-CAP in the fully streaming model, and present matching space lower bounds and upper bounds assuming our weighted spanner result.
In Section~\ref{sec:spanner}, we present our weighed spanner algorithm in the streaming model with better $\log W$ dependence.
Finally in Section~\ref{sec:further}, we present further applications to other network design problems such as $k$-ECSS and SNDP.

\section{Connectivity Augmentation in Link Arrival Streams}
\label{sec:linkarrive}
In this section, we consider $k$-CAP, the problem of augmenting the connectivity of a given graph $G=(V,E)$ from $k-1$ to $k$ using a subset of weighted links $L \subseteq \binom{V}{2}$ in link arrival streams. To recall, in the link arrival model, a cactus representation of the graph $G$, which is of size $O(n)$ (see Definition~\ref{def:cactus} for the formal definition of cactus), is given to us in advance and the set $L$ arrives in the stream (see Section~\ref{sec:model}). 

\begin{theorem}\label{thm:link-arrival-alg}
The $k$-connectivity augmentation problem ($k$-CAP) on ($G =(V, E),L$) in the link arrival model admits a one-pass $(2+\eps)$-approximation algorithm with total memory space $O(\frac{n}{\eps} \log \min(n,W))$ words, where $W= \max_{e\in E} w(e)$. 
\end{theorem}
Note that the augmentation set itself may have size $\Omega(n)$\footnote{As an example, consider a graph $G=(V,E)$ where $V=\{0,1,\dots,n-1\}$ and $E=\{(i,j): j-i\in \{1,2,\dots,k\}\}$ (where indices are modulo $n$), which has edge connectivity $2k$. If the link set is $L=\{(i,i+1):i\in [n]\}$, then at least $\lceil n/2\rceil$ links are necessary to increase the edge connectivity of $G$ to $2k+1$.}, so any algorithm that explicitly stores an augmentation solution in memory must have space complexity $\Omega(n)$. Moreover, we will show in Proposition~\ref{prop:lb-trivial-link-arrival} that just approximating the optimal total weight of the augmentation solution to any factor already requires $\Omega(n)$ bits of space. Hence, the space complexity of our algorithm is tight up to a poly-logarithmic factor. 
\textcolor{blue}{
}
Next, we describe our algorithm for $k$-CAP.

\subsection{Preliminaries}
\paragraph{Cactus representation of min-cuts.}
To increase the edge-connectivity of a $(k-1)$-connected graph $G$ to $k$, we need to add links to {\em cover} all min-cuts of size $k-1$. That is, for each cut $S$ of size $k-1$ (i.e., $|\delta_{G}(S)| = k-1$), we must add a link $e\in L$ such that $e\in \delta(S)$. \citet{DinitsKL76} showed there is a compact representation of all min-cuts of an undirected graph by a {\em cactus graph}.

\begin{definition}[Cactus Graph]\label{def:cactus}
A cactus graph is a $2$-edge-connected graph $C=(V_C, E_C)$ where each edge in $E_C$ belongs to exactly one simple cycle. Note that we allow cycles of length $1$ or $2$ too.
\end{definition} 

\begin{lemma}[\cite{DinitsKL76}]\label{lem:cactus-rep}
Let $G=(V,E)$ be an undirected graph. There is a loopless cactus $C=(V_C,E_C)$ of size at most $2n-1$ and a mapping $\varphi: V\rightarrow V_C$ so that a subset $S \in V$ is a min-cut of $G$ if and only if $\varphi(S)$ is a min-cut of C. 

Moreover, when the min-cut size of $G$ is an odd integer, the cactus representation of $G$ is a spanning tree (we may still treat it as a cactus by duplicating each tree edge).
\end{lemma}

The cactus representation is particularly useful for connectivity augmentation problems:
\begin{corollary}\label{cor:cactus}
Let $G=(V,E)$ be an undirected graph. Let $C$ denote the cactus representation of min-cuts in $G$. Then a link $(u,v)\in E\setminus E_H$ crosses a min-cut $S$ in $G$ if and only if the corresponding link $(\varphi(u),\varphi(v))$ crosses $\varphi(S)$ in $C$.
\end{corollary}

\begin{remark}
We remark that there is a simple streaming algorithm for constructing the cactus representation with space complexity $\tilde O(kn)$: First, construct a $k$-connectivity certificate $H$ of $G$ (recall that a $k$-connectivity certificate for a graph $G$ is a subgraph $H$ of $G$ that contains all edges crossing cuts of size $k$ or less in $G$~\citep{NagamochiI92}) with $O(kn)$ edges with space complexity $O(kn)$ words in polynomial time, using a simple algorithm by~\citet{NagamochiI92} (see also Lemma~\ref{lem:kcert}). Then, we apply the algorithm of~\citep{soda/KargerP09} for computing the cactus representation of the subgraph $H$ in $\tilde O(|E(H)|) = \tilde O(kn)$ time and space. It is straightforward to verify that the constructed cactus is a cactus representation of $G$, given $G$ is $(k-1)$-connected graph.

We then get the following as a corollary of Theorem \ref{thm:link-arrival-alg}:
 If the algorithm receives a $k$-connectivity certificate as a representation of $G$ or the edges of $G$ arrive in the stream before any link arrives, we can construct a cactus representation of $G$ in $O(kn)$ space first and then run our algorithm in this section for cactus augmentation and the overall space complexity will be $O(nk + \frac{n}{\eps}\log n)$.    
\end{remark}

\paragraph{Transforming cactus to cycle.} In the (weighted) {\em cactus augmentation} problem, without loss of generality, we can assume the cactus is a single cycle. The latter problem is known as weighted {\em cycle augmentation}. To reduce an instance on a general cactus to the single cycle case (without losing approximation factor), we apply the technique observed in~\citep{galvez2021cycle,traub20231}: Unfold the cactus into its Eulerian circuit, then add additional zero-weight edges (which we can use to augment at no cost) to connect the nodes corresponding to the same junction node in the cactus. See Section~3 in~\citep{traub20231} for a detailed description. 
\begin{lemma}[Theorem~3 in~\citep{galvez2021cycle}; see also Lemma~2.2 in~\citep{traub20231}]
Let $\alpha > 1$. If there is an $\alpha$-approximation algorithm for the weighted cycle augmentation problem, then the weighted cactus augmentation problem admits an $\alpha$-approximation.
\end{lemma}
Note that this reduction only produces $O(n)$ extra zero-weight edges, so it does not affect the space complexity of the streaming algorithm. We can apply the unfolding technique in the preprocessing step and in the rest of this section, we assume that the cactus is simply a single cycle.

\subsection{Main Step: Cycle Augmentation in Link Arrival Streams}
We arbitrarily assign a root node on the cycle, and let its index be $0$.  Then let the vertices of the cycle be $V=\{0,1,\dots,n-1\}$, with edges $C=\{e_1,e_2,\dots,e_n\}$ where $e_i = (i-1,i)$ (with indices modulo $n$).
We first describe a $2$-approximation for the unweighted case, using an idea from \citep{KhullerT93,KhullerV94}. 

\begin{theorem}\label{thm:unweighted}
There exists a one-pass $2$-approximation algorithm for the cycle
augmentation problem on unweighted graphs with total memory space $O(n)$ edges.
\end{theorem}
\begin{proof}
Following~\citep{KhullerT93,KhullerV94}, we consider a directed version of the problem defined as follows: given a set $E$ of \emph{directed} edges, augment a minimum size subset $E'\subseteq E$  to the cycle, such that for every $2$-cut $(L,V\setminus L)$ of the cycle where $0\in V\setminus L$ (i.e., $L =\{l,l+1,\dots,r\}$ for some $1\le l\le r\le n-1$), there exists $\vec{xy}\in E'$ with $y\in L$ and $x\in V\setminus L$ (we say $\vec{xy}$ covers $L$ in this case).
To reduce the original (undirected) cycle augmentation instance to this directed problem, simply replace each input edge $(u,v)$ by two arcs $\vec{uv},\vec{vu}$, incurring a $2$-factor approximation: any directed solution $\{\vec{xy}\}$ implies an undirected solution $\{(x,y)\}$ of the same cost, and any undirected solution $\{(x,y)\}$ implies a directed solution $\{\vec{xy}\}\cup \{\vec{yx}\}$ of twice the cost.
 
Now we solve the directed instance exactly by an $O(n)$-space streaming  algorithm. For each $v\in V$, we only need to keep the input arc $\vec{uv}$ with minimum indexed $u$, and keep the input arc $\vec{uv}$  with maximum indexed $u$.
In this way we store only $O(n)$ arcs in total, and finally we run an offline exact algorithm (e.g.,~\citep{Gabow95}, which was also used by~\cite{KhullerV94}) for the directed problem  
on these stored arcs.
This does not affect optimality, because when $0\le u<u'<v$, any $2$-cut, $U=\{l,l+1,\dots,r\}$ covered by $\vec{u'v}$ is also covered by $\vec{uv}$, so we can discard $\vec{u'v}$ if we already have $\vec{uv}$ (a similar argument applies to the $v<u'<u\le n-1$ case).
\end{proof}
By a simple scaling, this algorithm can be modified into a $(2+\eps)$-approximate algorithm for the weighted case with total space $O(\frac{n}{\eps}\log W)$ edges. Now we improve this $\log W$ dependency.

\begin{theorem}
The cycle augmentation problem on weighted graphs admits a one-pass $(2+\varepsilon)$-approximation streaming algorithm with total memory space $O(\frac{n}{\eps}\log \min(W, n))$ edges.
\end{theorem}
\begin{proof}
We assume $\eps>1/n$; otherwise use the trivial $O(n^2)$-space algorithm that stores the cheapest edge between every pair of vertices.
    
   Define weight intervals $I_k = [(n/\eps)^k, (n/\eps)^{k+1})$.
   Let $E_k$ be the set of input edges $e$ that have arrived so far with weights $w(e)\in I_k$.
   Note that\footnote{This inequality is meaningful only if both $E_k$ and $E_{k+2}$ are nonempty. This issue does not affect our overall argument since our algorithm can simply ignore the empty weight classes.} 
   \begin{equation}
     \label{eqn:gap}
  \frac{\min_{e\in E_{k+2}}w(e)}{\max_{e\in E_k}w(e)}> n/\eps.
   \end{equation}
These weight intervals do not contain zero, so we separately use a zero-weight class $E_{-1}$ to hold edges of zero weight. But for notational simplicity, we will not specially mention this zero weight class in later description. One can check that this does not affect the correctness of the algorithm.
   
Recall $C$ is the base cycle of length $n$. For each $k\in \{0,1,\dots,\lceil \log_{n/\eps}W\rceil\}$, define graph 
   \begin{equation}
        G_k:= C \cup \bigcup_{i\ge 0}E_{k-2i}.
   \end{equation}
Let $Q_k$ denote the collection of 3-edge-connected components\footnote{Recall that vertices $u,v$ belong to the same 3-edge-connected component if and only if there exist 3 edge-disjoint paths between $u$ and $v$.} of $G_k$, which form a partition of  the $n$ vertices. See Figure~\ref{fig:3edge} for an illustration. Let $1\le |Q_k|\le n$ denote the number of components. Since $G_{k} \subseteq G_{k+2}$,  $Q_{k}$ refines $Q_{k+2}$, and $|Q_{k}|\ge |Q_{k+2}|$.

\begin{figure}[!h]
    \centering \includegraphics[width=.25\textwidth]{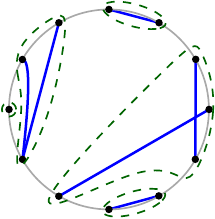}
    \caption{An example of 3-edge-connected components $Q_k$ of the graph $G_k$. Thin black edges denote the base cycle, and thick blue edges denote the links from the set $\bigcup_{i\ge 0}E_{k-2i}$; together they form $G_k$. The dashed green lines describe the 3-edge-connected components of graph $G_k$.}
    \label{fig:3edge}
\end{figure}

\paragraph{Algorithm description}  
At any point, our streaming algorithm always stores a subset of the input edges $E =\bigcup_{k} E_k$, which includes the following:
 \begin{enumerate}[leftmargin=*]
  \item \textbf{Undirected edges $F_k$:} We store edge subsets $F_k\subseteq E_k$, such that for all $k$ the subgraph $C \cup \bigcup_{i\ge 0}F_{k-2i} \subseteq G_k$ has the same 3-edge-connected components as $Q_k$. 
    \label{item:qk}
    \item \textbf{Directed arcs $S_k$:} For each $k$ and 3-edge-connected component $U\in Q_k$, and every weight interval $J_i= [(1+\eps)^{i},(1+\eps)^{i+1}) \subseteq I_{k+2}$, 
    we store the arc $\vec{xy}$ with minimum (and maximum) indexed $x$ where $(x,y)\in E_{k+2}$, $y\in U,x\notin U$, and $w(x,y)\in J_i$. The set of these arcs is denoted by $S_k$.
    \label{item:leftright}
 \end{enumerate}
   
Now we describe how to maintain this information when a new input edge $(u,v) \in E_{k'}$ arrives.
\begin{itemize}[leftmargin=*]
    	\item \textbf{Maintain Item~\ref{item:qk}:}
Note that adding this edge could potentially cause the components in $Q_{k'+2i}$ $(i=0,1,2,\dots)$ to merge.
To maintain Item~\ref{item:qk} (and hence the knowledge of all $Q_k$), we insert $(u,v)$ into the current $F_{k'}$, and then run a clean up procedure to remove redundant edges: Start from the graph $H\gets C \cup \bigcup_{j\ge 1}F_{k'-2j}$ which encodes the 3-connectivity information of the graph formed using edges prior to $E_{k'}$, 
and iterate over the edges $e\in F_{k' + 2i}$ (in increasing order of $i=0,1,2,\dots$). If adding $e$ to $H$ does not change the $3$-edge-connected components of $H$, then remove $e$ from $F_{k'+2i}$. Otherwise add $e$ to $H$.
It is clear that this clean up procedure preserves all the $3$-connectivity information, since we start from the base graph $C$ which is already 2-edge-connected.

	\item \textbf{Maintain Item~\ref{item:leftright}:}
 To maintain Item~\ref{item:leftright}, we simply use arcs $\vec{uv}$ and $\vec{vu}$ to replace the existing ones that become dominated. When two 3-edge-connected components $U,U' \in Q_k$ merge, we also merge the stored information for $U,U'$ (compare the best arcs stored for these two components and keep the better one).

	\item \textbf{Offline step:} In the end, we run 
an offline exact algorithm (such as \citep{Gabow95})
that solves the directed problem (see proof of Theorem~\ref{thm:unweighted}) on the stored arcs in Item~\ref{item:leftright} and directed versions of the stored edges in Item~\ref{item:qk}.
\end{itemize}

\paragraph{Space complexity.}
For Item~\ref{item:qk} the total space is $\sum_{k}|F_k| = \sum_{j}|F_{2j}|+\sum_{j}|F_{2j+1}|$ edges. We bound both terms separately. 
Due to our clean up procedure, there should be no redundant edges in $F_{\even}=\bigcup_j F_{2j}$: starting from the base cycle $H\gets C$, we can iterate over the edges $e\in F_{\even}$ in certain order so that adding edge $e$ to $H$ always strictly decreases the number of 3-edge-connected components of $H$. Hence $|F_{\even}|\le n-1$, and similarly $|F_{\odd}|\le n-1$, so 
\begin{equation}
  \label{eqn:boundf}
  \sum_k |F_k|\le 2(n-1).
\end{equation}

Define $c_{k+2}$ to be the number of 3-edge-connected components $U \in Q_{k}$ for which there exists $(u,v)\in E_{k+2}$ with $u\in U$ and $v\notin U$. Then the space for Item~\ref{item:leftright} is $\sum_{k}2c_{k+2}\cdot \log_{1+\eps}\frac{\max_{e\in E_{k+2}}w(e)}{\min_{e\in E_{k+2}}w(e)} \le \sum_{k}c_{k+2}\cdot O(\log(n/\eps)/\eps) \le \sum_{k}c_{k+2}\cdot O(\log(n)/\eps) $ edges.
We need the following lemma.
   \begin{lemma}
     \label{lem:diff}
 $c_{k+2} \le 2(|Q_{k}|-|Q_{k+2}|)$.
   \end{lemma}
   Using this lemma, the space complexity for Item~\ref{item:leftright} is
\begin{align*}
O(\frac{\log n}{\eps})\sum_{k}c_{k+2}  
&\le O(\frac{\log n}{\eps})\sum_{k}(|Q_k|-|Q_{k+2}|) \\
&\le O(\frac{\log n}{\eps})\cdot 2(n-1) \\
&\le O(\frac{n\log n}{\eps}) &&\hspace{-1.7cm}\rhd\text{by summing over even and odd $k$ separately}
\end{align*}
So the total space complexity is $O(\eps^{-1}n\log n)$ edges.

\begin{proof}[Proof of Lemma~\ref{lem:diff}]
We first shrink the graph $G_{k} = C \cup \bigcup_{i\ge 0}E_{k-2i}$ into graph $H_{k}$. Let each node of $H_{k}$ represent a 3-edge-connected component $U\in Q_{k}$, and for every $(u,v)\in C$ (recall $C$ is the set of edges on the base cycle) with $u\in U\in Q_{k}$ and $v\in V\in Q_{k}$, we connect $U,V$  in $H_{k}$ by an edge (allowing self-loops and parallel edges). As a standard fact, $H_{k}$ is a cactus (allowing self loops), and $C$ corresponds to an Eulerian circuit of $H_k$.

Let $\mathcal{C}_k$ denote the collection of simple cycles (cycles with distinct vertices; we view a self loop as a simple cycle as well) of the cactus $H_k$. 
Then $\mathcal{C}_k$ can be viewed as a partition of the $n$ edges on the base cycle $C$, where $e,e'\in C$ belong to the same partition if and only if $\{e,e'\}$ is a $2$-cut of $G_k$. 
Observe that $|\mathcal{C}_k| = n+1 - |Q_k|$.\footnote{To see this equality, consider removing one arbitrary edge from each simple cycle of the cactus, and the remaining edges should form a tree. As there are $|Q_k|$ vertices, the number of edges in the remaining tree is $|Q_k|-1$, so the number of edges $n$ in the original cactus equals $|Q_k|-1+|\mathcal{C}_k|$, since we removed $|\mathcal{C}_k|$ edges in the removal step.
} Hence, in the following it suffices to prove $c_{k+2} \le 2\cdot (|\mathcal{C}_{k+2}| - |\mathcal{C}_{k}|)$. Note that $\mathcal{C}_{k+2}$ is a finer partition of $C$ than $\mathcal{C}_k$.

\begin{figure}[!ht]
    \centering \includegraphics[width=.9\textwidth]{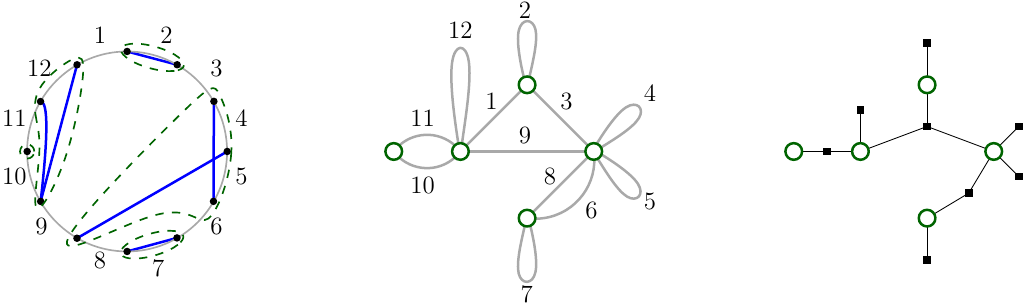}\caption{A picture of the cactus $H_k$ (middle) produced by shrinking $G_k$ (left). The tree $T_k$ (right) is produced from cactus $H_k$.}
    \label{3edge-reduction}
\end{figure} 
    
\begin{figure}[!h]
    \centering \includegraphics[width=.9\textwidth]{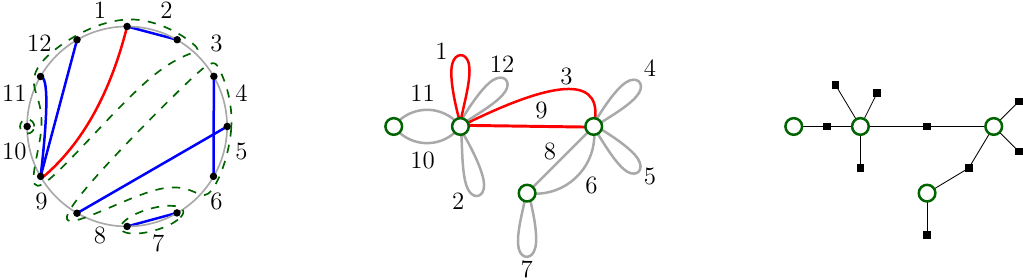}\caption{After adding an edge from $E_{k+2}$ (depicted in red), the partition $\mathcal{C}_{k+2}$ refines the old partition $\mathcal{C}_{k}$: $\{1,3,9\}$ breaks into $\{1\}$ and $\{3,9\}$.}
    \label{fig:3edge-after}
\end{figure}

We now consider how adding edges $E_{k+2}$ into $G_k$ can refine $\mathcal{C}_k$.
Convert the cactus $H_{k}$ into a tree $T_{k}$ as follows.
 Let $T_{k}$ be a bipartite graph with vertex bipartition $(\mathcal{C}_k,Q_k)$, in which $D\in \mathcal{C}_k$ is connected to every $U\in Q_k$ that lies on the simple cycle $D$ in cactus $H_{k}$.
Observe this bipartite graph $T_{k}$ is indeed a tree.  For each $(u,v)\in E_{k+2}$, let $u\in U\in Q_k$ and $v\in V\in Q_k$, and we mark all the tree-edges on the unique path connecting $U$ and $V$ in $T_k$.  For each $D\in \mathcal{C}_k$, let $d(D)$ denote the number of marked tree-edges incident to the tree-node $D$ in $T_k$. Then, observe that $d(D)\in \{0\}\cup \{2,3,4,\dots\}$, and $D \in \mathcal{C}_k$ (viewed as a subset of $C$) breaks into $\max\{d(D),1\}$ subsets in the partition $\mathcal{C}_{k+2}$.

By assumption, there are at least $c_{k+2}$ many tree-nodes $U\in Q_k$ that are incident to at least one marked tree-edge in $T_k$, so $T_k$ contains at least $c_{k+2}$ marked tree-edges. Hence, 
\begin{align*}
  |\mathcal{C}_{k+2}| &= \sum_{D\in \mathcal{C}_k}\max\{d(D),1\}\\
  & \ge \sum_{D\in \mathcal{C}_k}(1 + d(D)/2) &&\rhd\text{since $d(D)\in \{0\}\cup \{2,3,4,\dots\}$}\\
  & = |\mathcal{C}_k| + \frac{1}{2} \sum_{D\in \mathcal{C}_k}d(D)
  \ge |\mathcal{C}_k| + \frac{1}{2} c_{k+2},
\end{align*}
which completes the proof. 
\end{proof}
 
\paragraph{Approximation factor.}
Let $\opt\subseteq E=\bigcup_{k=-\infty}^{+\infty} E_k$ denote the optimal solution for the (undirected) cycle augmentation problem. Let $k^*$ be the maximum $k^*$ such that $\opt \cap E_{k^*} \neq \varnothing$. Then by \eqref{eqn:gap} we have 
 \begin{equation}
   \label{eqn:cheap}
  w(\opt)> (n/\eps) \cdot\max_{e\in E_{k^*-2}}w(e).
 \end{equation}
Bidirecting $\opt$ gives a solution $\opt'$ for the directed problem with total cost $w(\opt')= 2w(\opt)$. In the following we convert $\opt'$ into a solution $\sol$ for the directed problem that only uses arcs stored by the streaming algorithm,  with total cost $w(\sol)\le (1+O(\epsilon)) w(\opt') \le (2+O(\eps))w(\opt)$. This establishes that our streaming algorithm achieves $2+O(\eps)$ approximation ratio for the (undirected) cycle augmentation problem.
   
In $\sol$ we first include both directed versions of all $(u,v)\in \bigcup_{k\le k^*-2} F_{k}$, with total cost at most
 \begin{align*}
    \sum_{k\le k^*-2} 2|F_k|\cdot \max_{e\in F_k}w(e)
    & \le \sum_{k} 2|F_k| \cdot \max_{e\in E_{k^*-2}}w(e)\\
    & \le 4(n-1) \cdot \frac{\eps}{n} w(\opt) &&\rhd\text{by \eqref{eqn:boundf} and \eqref{eqn:cheap}}\\
    & \le 4\eps w(\opt).
 \end{align*}
Then, for every arc $\vec{xy}\in \opt'$ with weight $w(\vec{xy})\in I_k$ where $k\in \{k^*-1,k^*\}$,  we will find a replacement arc $\vec{x'y'}\in S_k$ stored by 
Item~\ref{item:leftright}:
Let $y\in U\in Q_k$. If $x\notin U$, then by Item~\ref{item:leftright} we can pick a stored arc $\vec{x'y'}\in S_k$ with $y'\in U$ and $w(\vec{x'y'}) < (1+\eps)w(\vec{xy})$, such that $x'\le x$ (if $x<y$) or $x'\ge x$ (if $x>y$). We include $\vec{x'y'}$ in $\sol$. (in the case of $x\in U$ we do not need to do anything)

By definition we immediately have $w(\sol) \le 4\eps w(\opt) + (1+\eps)w(\opt') = (2+6\eps)w(\opt)$.

To show $\sol$ is a feasible solution for the directed problem, we verify that each 2-cut $L = \{l,l+1,\dots,r\}$ (where $1\le l\le r\le n-1$) is covered. There are three cases:
\begin{itemize}[leftmargin=*]
  \item {\bf Case 1: $(L,V\setminus L)$ is not a $2$-cut of $G_{k^*-2}$.} By Item~\ref{item:qk}, $G_{k^*-2}$ and $C\cup \bigcup_{i\ge 0}F_{k^*-2-2i}$ have the same $3$-edge-connected components, and hence have the same $2$-cuts, so $(L,V\setminus L)$ is also not a $2$-cut of $C\cup \bigcup_{i\ge 0}F_{k^*-2-2i}$. Hence, there exists $(u',v')\in \bigcup_{i\ge 0}F_{k^*-2-2i}$ such that $u'\in V\setminus L, v'\in L$. Then,  $\vec{u'v'}$ covers $L$, and by construction we have $\vec{u'v'}\in \sol$.
  \item {\bf Case 2: $(L,V\setminus L)$ is not a $2$-cut of $G_{k^*-3}$.} This case is similar to case 1.
  
  \item {\bf Case 3: Otherwise.} In this case, $(L,V\setminus L)$ is a $2$-cut of both $G_{k^*-2}$ and $G_{k^*-3}$.
  
From the feasibility of $\opt$, we know there must exist arc $\vec{xy} \in \opt'$ that covers $L$ (i.e., $y\in L, x\in V\setminus L$) with weight $w(\vec{xy})\in I_k$ where $k\in \{k^*-1,k^*\}$. Let $y\in U\in Q_{k-2}$. Since $(L,V\setminus L)$ is a 2-cut of $G_{k-2}$, we know $x,y$ cannot be in the same $3$-edge-connected component of $G_{k-2}$, so $x\notin U$. Now let $\vec{x'y'}\in \sol$ be the replacement arc we found for $\vec{xy}$. By definition, $y'\in U$. We consider the case of $x<y$ (the other case $y<x$ is similar), and hence $x'\le x$. In this case we must have $x<l \le y \le r$, so $x'<l$ and hence $x'\notin L$.
Suppose for contradiction that $\vec{x'y'}$ does not cover $L$. Then we must have $y'\notin L$. But this would mean $(L,V\setminus L)$ is a 2-cut in $G_{k-2}$ separating $y$ and $y'$, contradicting the assumption that $y,y'\in U$ belong to the same 3-edge-connected component of $G_{k-2}$. This proves that the replacement arc $\vec{x'y'}\in \sol$ indeed covers $L$.
\end{itemize}
\end{proof}

\subsection{Lower Bound for Approximating Optimal \texorpdfstring{$k$}{k}-CAP Weight}
\label{subsec:lblinkarrival}
Now we show that the approximation factor of our streaming algorithm is close to optimal.
\begin{theorem}\label{thm:lb-link-arrival}
    Any streaming algorithm that solves the weighted TAP in the link arrival model with better than $2$-approximation needs $\Omega(n^2)$ bits of space.
\end{theorem}
\begin{proof}
    Let the base tree contain $(2n+1)$ vertices $r, x_1,y_1,\dots,x_n,y_n$, and edges $(r,x_i),(x_i,y_i)$ for all $1\le i\le n$, where $n$ is an even number.
We reduce from the INDEX problem, where Alice's bit string is from $\{0,1\}^{\binom{[n]}{2}}$. For $1\le i<j\le n$,  she adds link $(x_i,x_j)$ to $L$ if and only if the $(i,j)$-th bit in her bit string is $1$. Each of these links has weight $1$. Then she sends the memory content of the streaming algorithm to Bob.

Suppose Bob wants to find out whether the $(i,j)$-th bit in Alice's bit string is $1$. He arbitrarily groups $[n]\setminus \{i,j\}$ into pairs $(a_1,b_1),(a_2,b_2),\dots,(a_{n/2-1},b_{n/2-1})$, and adds links $(y_{a_k},y_{b_k})$ of zero weight to $L$ for all $1\le k\le n/2-1$. Then he adds two links $(x_i,y_i),(x_j,y_j)$ of zero weight. Then Bob asks the streaming algorithm to output a better-than-$2$ approximation of the  optimal solution weight.

Observe that, in any feasible tree augmentation solution, in order to augment the tree edges $(x_i,y_i)$ for all $i$, one has to include all the $n/2+1$ links that Bob added (which have zero total weight). Then, it remains to cover the tree edges $(r,x_i)$ and $(r,x_j)$. If $(x_i,x_j)\in L$, then one can use link $(x_i,x_j)$ of weight $1$ to cover them. Otherwise, one needs to include two other links of the form $(x_i,x_{k'}),(x_j,x_{k''})$ where $k',k''\notin\{i,j\}$ to cover them.\footnote{To let the tree augmentation instance always have a feasible solution, we can assume Alice's bit string contains $1$ on positions $(1,i)$ for all $i$ (which does not affect the space bound asymptotically), so that we can always choose $k'=k''=1$ here to get a feasible solution of size $2$.} Hence, the optimal solution has total weight is $1$ if the $(i,j)$-th bit in Alice's bit string is $1$, and total weight at least $2$ otherwise.
Hence, a better-than-2-approximation streaming algorithm for TAP can be used to solve the INDEX problem, which requires at least $\Omega(n^2)$ space.
\end{proof}
\begin{figure}[!h]
    \centering \includegraphics[width=.5\textwidth]{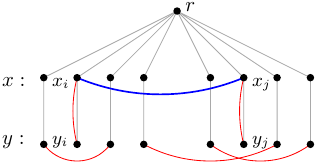}\caption{A lower bound instance of tree augmentation with weighted edges.
    The black edges form the base tree. The red links are added by Bob with zero weight. The blue link is added by Alice with weight $1$.}     \label{fig:weighted-augmentation-lb}
\end{figure} 

Note that the proof above also extends to $k$-connectivity augmentation for any value of $k$, by replacing each edge of the base tree by $k-1$ parallel edges.

We remark that the proof above crucially uses the fact that Bob's links have zero weight. In the unweighted setting, the same argument can only yield a lower bound for $1+O(1/n)$-approximate algorithms.  Nevertheless, we observe that a trivial reduction from the bipartite matching problem can yield a (non-tight) lower bound for approximating unweighted tree augmentation.
\begin{corollary}\label{col:lb-unweighted}
    There is no streaming algorithm with $O(n(\log n)^{O(1)})$ space that outputs a solution to the unweighted TAP in the link arrival model with better than $1.409$-approximation.
\end{corollary}
\begin{proof}
   We invoke a lower bound result by~\cite{DBLP:conf/soda/Kapralov21}, which states that in a bipartite graph $G$ 
   with $n$ vertices and maximum matching size $(1/2-O(\eps))n$, a single-pass semi-streaming algorithm cannot output a matching of $G$ of size larger than $(\frac{1}{2(1+\ln 2)}+\eps)n$.

   To reduce from bipartite maximum matching on a $n$-vertex graph to unweighted TAP on a $(n+1)$-vertex tree in the link arrival model, we assume the base tree is a star graph with edges $(r,x_1), \cdots, (r,x_n)$.  
   For every edge $(i,j)$ in the matching input instance, we create a link $(x_i,x_j)$ for the unweighted TAP instance. Then, we add links $(r,x_i)$ for all $i\in [n]$. Then, a bipartite matching $M=\{(i,j)\}$ implies a tree augmentation solution $\{(x_i,x_j)\}_{(i,j)\in M}\cup \{(r,x_k): \text{$k$ is unmatched in $M$}\}$, which consists of $|M|+n-2|M|=n-|M|$ links. Observe that the converse direction also holds. Hence, plugging in the result of \cite{GoelKK12}, when the optimal tree augmentation solution has size $(1/2+O(\eps))n$, a single-pass semi-streaming algorithm cannot output a solution of size smaller than $n-
   (\frac{1}{2(1+\ln 2)}+\eps)n$. This means the approximation ratio for unweighted TAP cannot be better than $\frac{n - (\frac{1}{2(1+\ln 2)}+\eps) n}{(1/2+O(\eps))n} \approx 1.409$.
\end{proof}

Now we show that the $n$-dependency of our streaming algorithm is also necessary.
\begin{proposition}\label{prop:lb-trivial-link-arrival}
    Any streaming algorithm that outputs an multiplicative approximation (to any factor) of the optimal total weight for the weighted TAP in the link arrival model needs $\Omega(n)$ bits of space.
\end{proposition}
\begin{proof}
    Let the base tree contain $(n+1)$ vertices $r, x_1,\dots,x_n$, and edges $(r,x_i)$ for all $1\le i\le n$.
We reduce from the INDEX problem, where Alice's bit string is from $\{0,1\}^{n}$. For $1\le i\le n$,  she adds link $(r,x_i)$ to $L$ if and only if the $i$-th bit in her bit string is $1$. Each of these links has weight $0$. Then she sends the memory content of the streaming algorithm to Bob.

Suppose Bob wants to find out whether the $i$-th bit in Alice's bit string is $1$. He adds links $(r,x_j)$ of zero weight to $L$ for all $j\neq i$. Then he adds link $(r,x_i)$ of weight $1$. Then Bob asks the streaming algorithm to output an approximation of the optimal solution weight. 

Observe that the optimal solution to augment the tree has $0$ weight if the $i$-th bit of Alice's bit string is $1$, and has weight $1$ otherwise. 
Hence, any multiplicative approximation streaming algorithm for weighted TAP can be used to solve the INDEX problem, which requires at least $\Omega(n)$ bits of space.
\end{proof}

\section{Connectivity Augmentation in the Fully Streaming Setting}
\label{sec:fullystream}

In this section, we first prove a space lower bound for $k$-CAP in the fully streaming model. Then, we show a streaming algorithm with nearly matching space complexity.

\subsection{Lowerbound for Estimating Connectivity Augmentation Cost}
Our main lower bound statement is the following.
\begin{theorem}\label{thm:lowerbound-combined}
For any constant integer $t\ge 1$,  the (unweighted) $k$-CAP (even when $k$ is known) in the fully streaming model requires space complexity $\Omega(kn + n^{1+1/t})$ bits (assuming the Erd\H{o}s's girth conjecture) to approximate the solution size to a factor better than $2t+1$.
\end{theorem}
It follows from combining two lower bound results Theorem~\ref{thm:lowerbound-t} and Theorem~\ref{thm:lb-fully-nk}.
\paragraph{Lower bound in terms of approximation factor ($t$).}
We first describe the space lower bound in terms of the approximation factor. 
As is standard in the spanner literature, the proof is based on high-girth graphs, but here we need to be more careful to make the connection between tree-augmentation and shortest paths.

\begin{theorem}\label{thm:lowerbound-t}
Consider the (unweighted) TAP where $E$ is the base tree and $L$ is the set of edges to augment, and $E \cup L$ arrive as a stream in an arbitrary order. 

For any constant integer $t\ge 1$, any (randomized) streaming algorithm $\mathcal{A}$ that can output the size of a better than $(2t+1)$-approximate solution requires $\Omega(\gamma(n,2t+1))$ bits of space, where $\gamma(n,2t+1)$ denotes the maximum possible number of edges in an $n$–vertex graph with girth $> 2t+1$.
\end{theorem}
\begin{remark}
It is known that $\gamma(n,2t+1) = \Omega(n^{1+\frac{1}{t}})$
for $t\in \{1,2,3,5\}$ \citep{wenger1991extremal}, and for all integers $t$ assuming Erd\H{o}s' girth conjecture \citep{erdosgirth}. Moreover, for all integers $t$, $\gamma(n,2t+1) = \Omega(n^{1+2/(3t-2-(t\bmod 2))})$
unconditionally~\citep{lazebnik1995new}. 
\end{remark}

\begin{proof}
Let $G$ be a fixed graph on $|V|=n$ vertices with girth $>2t+1$ and $|E(G)| = \gamma(n,2t+1)$ edges.  Consider the INDEX  problem: Alice has a bit string from $\{0,1\}^{E(G)}$, viewed as a subgraph $G'\subseteq G$. Alice sends a message to Bob. Then, Bob needs to recover the $i$-th bit of the string for a given index $i$, or equivalently, decide whether $(u,v)\in G'$ for a given edge $(u,v) \in G$.
It is well known \citep{jcss/MiltersenNSW98} that any bounded-error randomized protocol for this task requires message size $\Omega(|E(G)|)$ bits. 

Now we use the streaming algorithm $\mathcal{A}$ for TAP to design a protocol for the INDEX problem:
Alice and Bob together will construct a TAP instance $(E,L)$ for $\mathcal{A}$.
First, Alice feeds $L:=E(G')$ to $\mathcal{A}$, and then sends the current memory content of $\mathcal{A}$ to Bob. In order to determine whether $(u,v)\in G'$, Bob constructs a chain $E:=\{(x_1,x_2),(x_2,x_3),\dots,(x_{|V|-1},x_{|V|})\}$ and feeds $E$ to $\mathcal{A}$,  where the endpoints are $x_1:=u, x_{|V|}:= v$, and the remaining vertices $\{x_2,\dots,x_{|V|-1}\}=V\setminus \{u,v\}$ are sorted so that $d_{H}(u,x_{j-1}) \le d_{H}(u,x_{j})$ $(2\le j\le |V|-1)$, where graph $H$ is defined to be $G$ with edge $(u,v)$ removed. Then, Bob decides $(u,v)\in G'$ if and only if $\mathcal{A}$ reports an approximate answer $< 2t+1$ on the instance $(E,L)$. Now we show the correctness of this protocol.
\begin{itemize}[leftmargin=*]
  \item {\bf Case $(u,v)\in E(G')$:} In this case, the optimal solution for augmenting the chain 
  \[E=\{(u,x_2),(x_2,x_3),\dots,(x_{|V|-1},v)\}\] 
  is to include the single edge $(u,v)\in L=E(G')$ which completes a cycle. Hence, $\mathcal{A}$ should report an approximate answer $<2t+1$. So Bob correctly decides $(u,v)\in E(G')$.

\item {\bf Case $(u,v)\notin E(G')$:} In this case, in order to show Bob correctly decides $(u,v)\notin E(G')$, it suffices to show that any feasible augmentation solution $S=\{(x_{i_1},x_{j_1}),(x_{i_2},x_{j_2}),\dots,$ $(x_{i_s},x_{j_s})\}\subseteq L=E(G')$ must have size at least $s\ge 2t+1$.

We assume $i_k<j_k$ for all $1\le k\le s$. Since $\{(x_1,x_2),(x_2,x_3),\dots,$ $(x_{|V|-1},x_{|V|})\}\cup S$ is $2$-edge-connected, we know the intervals $[i_1,j_1-1],\dots,[i_s,j_s-1]$ covers all $\{1,2,\dots,|V|-1\}$.  We can assume $1=i_1<i_2<\dots < i_s$ and $j_1<\dots < j_s=|V|$ and $j_{k-1}\ge i_{k}$ without loss of generality (by keeping a minimal feasible subset of $S$). 

Note that $S\subseteq E(G') \subseteq E(G) \setminus \{(u,v)\} = E(H)$.
Now we inductively prove for every $1\le k \le s$ that $d_H(u,x_{j_k})\le k$.   The base case $k=1$ is immediate: $d_H(u,x_{j_1}) = d_H(x_{i_1},x_{j_1}) = 1$. For the inductive step $2\le k\le s$, we have
\begin{align*}
    d_H(u,x_{j_k}) 
    &\le d_H(u,x_{i_{k}})+d_H(x_{i_k},x_{j_k})\\
    & = d_H(u,x_{i_{k}}) + 1\\
    & \le d_H(u,x_{j_{k-1}}) +1  &&\hspace{-1.1cm}\rhd\text{by $i_k\le j_{k-1}<|V|$ and monotonicity of $d_H(u,x_{*})$}\\
    & \le k &&\hspace{-1.1cm}\rhd\text{by induction hypothesis}
\end{align*}
Hence, this establishes that $d_H(u,v) = d_{H}(u,x_{j_s})\le s$. So $G = H \cup \{(u,v)\}$ contains a cycle of length $\le s+1$.  Since $G$ has girth $\ge 2t+2$, we conclude $s \ge 2t+1$, which finishes the proof.
\end{itemize}
\end{proof}
\begin{remark}
    The same lower bound  of Theorem~\ref{thm:lowerbound-t} also generalizes to $k$-CAP for higher values $k> 2$ (if we allow the base graph $E$ to have parallel edges): we simply replace each edge in Bob's chain $E$ by $(k-1)$ parallel edges, so that it becomes a $(k-1)$-edge-connected graph, and the rest of the proof works similarly.
\end{remark}

\paragraph{Lower bound in terms of connectivity parameter ($k$).}~\citet{zelke2011intractability} gave a simple proof that computing the size of the  minimum cut of an (unweighted) undirected graph requires $\Omega(n^2)$ bits of space for any one-pass streaming algorithm. In Zelke's construction the input graph has minimum cut size as large as $\Theta(n)$. Here we observe that Zelke's proof can be adapted to graphs with minimum cut size $\Theta(k)$, and show lower bounds for the connectivity augmentation problem.

We remark that \citep{sun2015tight} also obtained an $\Omega(kn)$-bit randomized lower bound and an $\Omega(kn\log n)$-bit deterministic lower bound for the $k$-CAP using a different proof.

\begin{theorem}\label{thm:lb-fully-nk}
 The $k$-CAP (where $k$ is known) in the fully streaming model (with unweighted links) requires $\Omega(nk)$ bits of space to approximate to any finite factor.
\end{theorem}
\begin{proof}
The proof is a straightforward adaptation of \citep{zelke2011intractability}.
Let $G$ be a fixed $k'$-regular graph on $|V|=n$ vertices and $|E(G)|=k'n/2$ edges. Following \citep{zelke2011intractability}, we give a reduction from the following INDEX problem:  
Alice has a bit string from $\{0,1\}^{E(G)}$, viewed as a subgraph $G'\subseteq G$. Alice sends a message to Bob. Then, Bob needs to recover the $i$-th bit of the string for a given index $i$, or equivalently, decide whether $(u,v)\in G'$ for a given edge $(u,v) \in G$.
It is well known \citep{jcss/MiltersenNSW98} that any bounded-error randomized protocol for this task requires message size $\Omega(|E(G)|) = \Omega(nk')$ bits. 

Now we use the streaming  algorithm $\mathcal{A}$ for $k$-CAP to design a protocol for the INDEX problem.
Alice and Bob together will construct a $k$-CAP instance $(G'',L)$ for $\mathcal{A}$, as follows: 
\begin{itemize}[leftmargin=*]
    \item Alice feeds her input graph $G'$ to the streaming algorithm $\mathcal{A}$, and sends the memory content of $\mathcal{A}$ as well as all degrees $deg_{G'}(u)$ for all $u\in V(G')$ to Bob. The latter has $O(n\log k')$ total bit length as $G'$ has maximum degree at most $k'$.
    \item Then, Bob builds graph $G''$ by adding vertices and edges to $G'$ (and feeding them to $\mathcal{A}$), as follows: First, adds two disjoint cliques  $K,K'$ each of size $3 k'$ (on vertices disjoint from $V(G')$). Then, adds edges between $(u,q),(v,q)$ for all $q\in K$, and edges $(w,q')$ for all $w\in V(G')\setminus \{u,v\}$ and all $q'\in K'$. Then, adds a special new vertex $c$, and connect it to $k-1:=deg_{G'}(u)+deg_{G'}(v)-2 \le 2k'-2$ vertices of $K'$.
    This finishes the construction of the graph $G''$.
    \item For two arbitrary vertices $r\in K, r'\in K'$, Bob adds two links $L:=\{(c,r'),(r',r)\}$.  Then Bob obtains the solution  of algorithm $\mathcal{A}$ on instance $(G'',L)$, where the connectivity should be increased from $k-1=deg_{G'}(u)+deg_{G'}(v)-2$ to $k$. 
    If the result is $1$, then Bob decides $(u,v)\notin G'$. Otherwise, Bob decides $(u,v)\in G'$. 
\end{itemize}

\begin{figure}[!h]
    \centering \includegraphics[width=.55\textwidth]{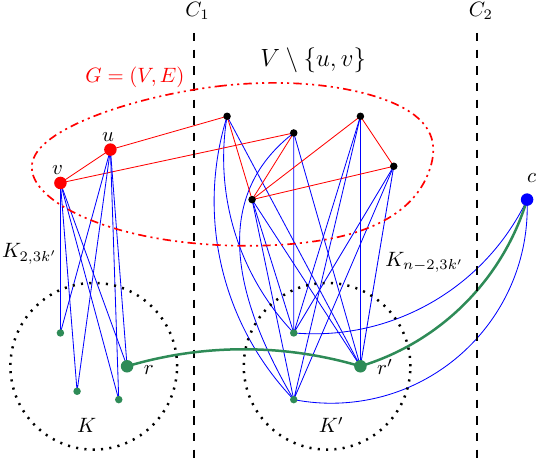}
    \caption{Illustration for the proof of Theorem~\ref{thm:lb-fully-nk}. The red edges are from Alice. The blue edges are from Bob. Then, the two green links are added by Bob.}
    \label{fig:nk-lb}
\end{figure}

To see the correctness of this protocol, we first analyze the minimum cuts of $G''$. Since $deg_{G''}(c) = k-1$, the edge connectivity of $G''$ is at most $k-1$. 
Then, any cut of $G''$ that separates the clique $K$ (or $K'$) must have size at least $\min_{1\le j< |K|}j(|K|-j) = 3k'-1> k-1$, and is therefore not a minimum cut.
Furthermore, any cut that separates $K$ from either of $\{u,v\}$ (or separates $K'$ from any of $V\setminus \{u,v\}$) must have cut size at least $|K|=3k'>k-1$, and is therefore not a minimum cut. This means $V(G'')$ is partitioned to three groups: $K\cup \{u,v\}$, $K'\cup V(G')\setminus \{u,v\}$, and $\{c\}$, such that none of the groups can be separated by any minimum cut. Since $K\cup \{u,v\}$ is not adjacent to $c$, we see that there are only two candidates $C_1=(V_1,V(G'')\setminus V_1),C_2=(V_2,V(G'')\setminus V_2)$ for minimum cut of $G''$, where $V_1= K\cup \{u,v\}$ and $V_2 = \{c\}$. By inspecting the construction of $G''$ we can see $|C_2|=deg_{G''}(c) = k-1$, and
$|C_1| = |E(K\cup \{u,v\}, K'\cup V(G')\setminus \{u,v\})| = |E(\{u,v\}, V(G')\setminus \{u,v\})| = deg_{G'}(u)+deg_{G'}(v) -2|E(\{u\},\{v\})| = \begin{cases} k+1 & (u,v)\notin E(G')\\k-1 & (u,v)\in E(G') \end{cases}$.
Hence, the edge connectivity of $G''$ is $k-1$, and there are two minimum cuts $C_1,C_2$ if $(u,v)\in E(G')$, and exactly one minimum cut $C_2$ if $(u,v)\notin E(G')$.

Now, observe that the link $(c,r')\in L$ can cover $C_2$ but not $C_1$, and the link $(r',r)\in L$ can cover $C_1$ but not $C_2$. So the connectivity augmentation solution has size $1$ if $(u,v)\notin E(G')$, and has size $2$ if $(u,v)\in E(G')$. Hence, any if $\mathcal{A}$ has approximation better than $2$, then this protocol correctly solves the INDEX problem. The message size of the protocol is the space complexity of $\mathcal{A}$ plus $O(n\log k')$ bits. Hence the space complexity of $\mathcal{A}$ must be at least $\Omega(k'n) - O(n\log k') \ge \Omega(k'n) = \Omega(kn)$, provided that $k'$ is larger than some constant.  The proof can also be adapted to any finite approximation ratio, by adding multiple links of the form $(r',r)\in K'\times K$ instead of just one.

It still remains to show the theorem for small constant $k$. Note that the $1$-connectivity augmentation problem (increasing connectivity from $0$ to $1$) is at least as hard as the problem of deciding whether a $n$-vertex graph is connected, and hence requires $\Omega(n\log n)$ bits of space by \cite{sun2015tight}. 
   \end{proof}
   \begin{remark}
      We remark that the same $\Omega(nk)$ lower bound also holds for the task of constructing a cactus representation of  a graph (Lemma~\ref{lem:cactus-rep}), even assuming the edge connectivity value $k$ is known. This is because the cactus representation immediately allows to distinguish between the cases of having two minimum cuts $C_1,C_2$ or one minimum cut $C_2$, and thus the proof above still applies.
   \end{remark}

\subsection{Tight Algorithm} \label{subsec:algo-fully}
 Next, we describe our single-pass algorithm that outputs a $(2t-1+\eps)$-approximate solution using $O(nk + \eps^{-1}n^{1+1/t}\log n)$ words of space, nearly matching the lower bounds of Theorem~\ref{thm:lowerbound-t} and~\ref{thm:lb-fully-nk}.

\paragraph{Compact $k$-connectivity certificate for edges}
A $k$-connectivity certificate, or simply a $k$-certificate, for an $n$-vertex graph $G$ is a subgraph $H$ of $G$ that contains all edges crossing cuts of size $k$ or less in $G$. Such a certificate always exists with $O(kn)$ edges, and moreover, there are graphs where $\Omega(kn)$ edges are necessary~\citep{goel2010graph}.
The following observation is folklore.
\begin{lemma}[Follows from~\citep{NagamochiI92}]
\label{lem:kcert}
   There is a one-pass streaming algorithm that computes a $k$-connectivity certificate of an $n$-vertex graph with at most $k(n-1)$ edges with space complexity $O(nk)$ words.
\end{lemma}
\begin{proof}
 Nagamochi and Ibaraki showed that the following procedure produces a $k$-edge connectivity certificate of $G=(V,E)$:  Iteratively for $i\in \{1,2,\dots,k\}$, let $F_i$ be any maximal forest of $E\setminus (F_1\cup \dots \cup F_{i-1})$, and finally output $F_1\cup \dots \cup F_k$ as a $k$-edge connectivity certificate of size at most $k(n-1)$. 
 
 This procedure can be easily implemented in the streaming model: maintain $k$ forests $F_1,\dots,F_k$ initially empty. For each arriving edge $e$, add $e$ to $F_i$ where $i$ is the smallest index such that $F_i\cup\{e\}$ is still a forest (discard $e$ if such $i$ does not exist).
\end{proof}
   
Next, we will prove the following theorem which gives a single-pass algorithm that outputs a $(2t-1+\eps)$-approximate solution using $O(nk + \eps^{-1}n^{1+1/t}\log n)$ words of space, nearly matching the lower bounds of Theorem~\ref{thm:lowerbound-t} and~\ref{thm:lb-fully-nk}.
\begin{theorem}
\label{thm:fullystreamingkconn}
The $k$-CAP in the fully-streaming model can be solved by a single-pass streaming algorithm with approximation ratio $(2t-1+\eps)$ and space complexity $O(nk + \eps^{-1}n^{1+1/t}\log n)$ words.
\end{theorem}
\begin{proof}
    Let $V$ denote the vertex set. We receive edges from $E$ and weighted links from $L$ in an arbitrary order, and we deal with $E$ and $L$ with two separate streaming algorithms: we build a $k$-connectivity certificate $E'\subseteq E$ of the graph $(V,E)$ using Lemma~\ref{lem:kcert} in $O(nk)$ words of space, and we build a $(2t-1+\eps)$-spanner $L'\subseteq L$ of the graph $(V,L)$ in $O(\eps^{-1}n^{1+1/t}\log n)$ words of space (Theorem~\ref{thm:spanner-alg}). 
    Finally, we use an exact  algorithm (by brute force in exponential time)     to solve the $k$-connectivity augmentation problem on instance $(E',L')$. Clearly, the space complexity of this algorithm is $O(nk + \eps^{-1}n^{1+1/t}\log n)$ words. It remains to prove the correctness of the algorithm:
    \paragraph{Feasibility.} Since $E$ is $(k-1)$-edge-connected, its $k$-certificate  $E'\subseteq E$ is also $(k-1)$-edge-connected. Given a feasible solution $L''\subseteq L'$ of the $k$-CAP instance $(E',L')$, we know $L''$ covers all $(k-1)$-cuts of $E'$, which are precisely all $(k-1)$-cuts of $E$ by definition of $k$-certificate $E'$, so $L''$ is a feasible solution for the input $k$-CAP instance $(E,L)$ as well.
    \paragraph{Approximation.} Given an optimal solution $L''\subseteq L$ to the input instance $(E,L)$, we replace every link $(u,v)$ in $L''$ by a path from $u$ to $v$ of length $\le (2t-1+\eps)w_{(u,v)}$ using links from the spanner $L'$. Let the union of this replacing links be $\tilde L''\subseteq L'$. Then it is clear that $\tilde L''$ is a feasible solution for $k$-CAP with total weight $w(\tilde L'')\le (2t-1+\eps)w(L'')$ as it covers all the $(k-1)$-cuts.
\end{proof}

\section{Streaming Algorithm for Spanners on Weighted Graphs}\label{sec:spanner}
In this section, we prove the following theorem on computing spanners for weighted graphs in the streaming model.
\begin{theorem}\label{thm:spanner-alg}
For any integer $t\ge 1$, there is a one-pass streaming algorithm for computing a $(2t-1+\eps)$-spanner of size $O(\eps^{-1}n^{1+1/t}\log n)$ of a weighted graph,  with space complexity $O(\eps^{-1}n^{1+1/t}\log n)$ words.
\end{theorem}

Let $G=(V, E)$ be a weighted graph. We denote the weight function by $w:E\to \mathbb{R}^+$. Moreover, We normalize the weights so that $w(e)\in \{0\}\cup [1,W]$.
For each $j \in [0, \ceil{\log_{1+\eps} W}]$, our algorithm stores $E_j$, a subset of edges of $G$ that have weights in $[(1+\eps)^j, (1+\eps)^{j+1})$.
(These intervals do not contain zero, so we separately use a zero-weight class $E_{-1}$ to hold edges of zero weight.  But for notational simplicity, we will not specially mention this zero weight class in later description. One can check that this does not affect the correctness of the algorithm.)

Our algorithm is as follows (see Algorithm~\ref{alg:1}). As an edge $e$ arrives, round its weight to the nearest power of $(1+\eps)$ and place it in the corresponding weight class $E_j$.
As usual, we keep the edge $e$ iff it does not close a cycle of length at most $2t$ in $E_j$, for some given parameter $t$.
After processing the edge, we run the~\textsc{Sparsify} subroutine described below in Algorithm~\ref{alg:sparsify}.

\paragraph{\textsc{Sparsify} subroutine.} Let $C>0$ be a sufficiently large constant. Define intervals $I_k=[k\cdot (C/\eps)\log n, (k+1)\cdot (C/\eps)\log n]$. For all $k$ let $\tilde E_k:=\bigcup_{j\in I_k} E_j$. For each $k$, let 
\begin{align*}
    E^{\even}_{\leq k}=\bigcup_{j=-\infty}^{k} \tilde E_{2j} \quad\text{ and }\quad    E^{\odd}_{\leq k}=\bigcup_{j=-\infty}^{k} \tilde E_{2j+1}.
\end{align*}

Let $E^{\even}=\bigcup_{j} \tilde E_{2j}$, and we define $E^{\odd}$ similarly.
Our \textsc{Sparsify} procedure operates independently on these two sets. We will ensure that each set contains $O(\eps^{-1}n^{1+1/t}\log n)$ edges, independent of the weight bound $W$.

We now describe how \textsc{Sparsify} operates on $E^{\even}$ (the operations are the same for $E^{\odd}$).

\begin{claim}
Let the constant $C$ in the definition of the sets $\tilde E_k$ be chosen sufficiently large. Let $k$ be an integer. Let $H=(V, E^{\even}_{\leq k-1})$. Then for any edge $e=(u, v)\in \tilde E_{2k}$ such that $u$ and $v$ belong to the same connected component in $H$, one has $w_{e}\geq \text{dist}_{H}(u, v)$.
\end{claim}
\begin{proof}
Since $w_e\in \tilde E_{2k}$, one has $w_e\geq (1+\eps)^{2k (C/\eps)\log n}$ by definition of $\tilde E_{2k}$.

On the other hand, the longest 
edge in $E^{\even}_{\leq k-1}$ has length at most  $(1+\eps)^{(2(k-1)+1) (C/\eps)\log n}$ by our definitions.
Thus, the ratio of length of shortest edge in $\tilde E_{2k}$ and length of longest edge in $E^{\even}_{\leq k-1}$ is at least $(1+\eps)^{(C/\eps)\log n}=n^{\Omega(C)}>n$ for sufficiently large constant $C$. The shortest path in $H$ has at most $n$ edges, so the claim follows.
\end{proof}

Our procedure \textsc{Sparsify}$(k)$ performs the following step for each $k$ from $k_{\max}$ down to $k_{\min}$. Collapse the connected components induced by $E^{\even}_{\leq k-1}$ into supernodes, and consider the multigraph with edges $\tilde E_{2k}$ on this set of supernodes. We convert this multigraph into a simple graph in the following natural way. For each edge $e=(u, v)\in \tilde E_{2k}$, 
\begin{itemize}
\item delete $e$ if it is a self loop in this graph (i.e. $u, v$ belong to the same connected component) 
\item delete $e$ if there is a shorter edge that is parallel to $e$.
\end{itemize}
This is summarized in Algorithm~\ref{alg:sparsify}.
\begin{algorithm}[H]
\caption{\textsc{Sparsify}}
\label{alg:sparsify} 
\begin{algorithmic}[1] 
\Procedure{Sparsify}{}
    \For{$k=k_{\max}$ down to $k_{\min}$}
        \State  Let $H=(V, E^{\even}_{\leq k-1})$
        \State  Let $C_1,\ldots, C_r$ be the connected components of $H$.
    
        \For {$e=(u, v)\in \tilde E_{2k}$ }
            \If{$u, v\in C_i$ for some $i$}
                \State delete $e$ from $\tilde E_{2k}$
            \EndIf
            \If{$\exists(u', v')\in \tilde E_{2k}$ s.t.\ $w_{(u',v')}\le w_{e} $, and $u, u'\in C_i$, $v, v'\in C_j$ for some $i$, $j$}
                \State delete $e$ from $\tilde E_{2k}$
            \EndIf
        \EndFor
    \EndFor
    \State $\rhd$ \emph{The same procedure for the set $E^{\odd}$}
\EndProcedure 
\end{algorithmic}
\end{algorithm}
The algorithm is summarized in Algorithm~\ref{alg:1}:

\begin{algorithm}[H]
\caption{Overall algorithm}
\label{alg:1} 
\begin{algorithmic}[1] 
\Procedure{Spanner}{}
\For{each edge $e=(u, v)$ in the stream}
\State Round weight of $e$ to power of $1+\eps$. Let $j$ be the weight class of $e$.
\State Add $e$ to $E_j$ iff $\text{dist}_{E_j}(u, v)>(2t-1)\cdot w_e$.\label{line:4}
\State Call \textsc{Sparsify}
\EndFor
\EndProcedure 
\end{algorithmic}
\label{algo:spanner}
\end{algorithm}

\begin{lemma}
\label{lem:spannerstretch}
The edges stored by Algorithm~\ref{alg:1} form a $(2t-1)\cdot (1+\eps)$-spanner of $G$.
\end{lemma}
\begin{proof}
It suffices to show that for each $e=(u,v)$ of $G$ that is not stored by Algorithm~\ref{alg:1}, there is a path  of length at most $(2t-1)\cdot (1+\eps)\cdot w_e$ connecting $u,v$ using the stored edges.

Let's consider when an edge $(u, v)$ is not included in the graph. The edge may be ignored upon arrival if its endpoints are connected by a path of length at most $(2t-1)\cdot w_e$ in the graph $E_j$ (at Line~\ref{line:4} in Algorithm~\ref{algo:spanner}), where $j$ is the weight class of $e$.  So in this case we have a short path to substitute for that edge (where we lose a $(1+\eps)$ factor due to rounding). The worry is that deleting edges in the call to \textsc{Sparsify} may break this argument, but it does not: we only delete an edge $e'\in E_{j'}$ in \textsc{Sparsify} if we have an even shorter path in $\bigcup_{i<j'} E_i$ to connect its endpoints. 
\end{proof}

\begin{lemma}
\label{lem:spannerspacebound}
Throughout the algorithm, the total number of edges stored by Algorithm~\ref{alg:1} is always at most  
$O(\eps^{-1}n^{1+1/t}\log n)$.
\end{lemma}
\begin{proof}
We analyze the even case, and the odd case can be analyzed similarly.

For each $k$, let $c_k$ denote the number of connected components in $E^{\even}_{\leq k-1}$.
For each $j\in I_{2k}$, we view $E_j$ as a graph on the supernodes formed by the $c_k$ connected components of $E^{\even}_{\leq k-1}$. 
Then, the \textsc{Sparsify} function maintains the property that each graph $E_j$ forbids cycles of at most $2t$ edges.
Let $0\le d_k \le c_k$ denote the number of supernodes that have non-zero degree in at least one of these graphs $E_j$. Then, the number of edges in each $E_j$ is at most $O(d_k^{1+1/t})$ edges by the Moore bound (see e.g., \citep{DBLP:journals/dcg/AlthoferDDJS93}). Summing up all $2(C/\eps)\log n$ many indices $j\in I_{2k}$, we know $\tilde E_{2k}$ has at most $O(\eps^{-1}d_k^{1+1/t}\log n)$ edges.

Note that $c_k$ is monotonically non-increasing in $k$, and we claim that $c_{k} - c_{k+1} \ge d_k/2$.
To show this claim, we take a spanning forest of the edges in $\tilde E_{2k}$, and suppose this spanning forest has $m'$ edges. Then $m' = c_{k}-c_{k+1}$. And $d_k$ equals the number of nodes connected to at least one of these $m'$ edges. If we charge each of these nodes to one of its incident edges, then each edge is charged at most twice, showing that $d_k\le 2m'$ as claimed.

This claim implies that $\sum_{k} d_k \le 2\sum_{k}(c_{k}-c_{k+1}) \le 2n$. Hence, the total number of stored edges is 
\[ \sum_{k}|\tilde E_{2k}|\le \sum_{k}O(\eps^{-1}d_k^{1+1/t}\log n)\le O(\eps^{-1}n^{1/t}\log n)\cdot \sum_{k}d_k \le O(\eps^{-1}n^{1+1/t}\log n).\]

Note that this lemma implies both the space complexity bound and the sparsity bound claimed and completes the proof of Theorem~\ref{thm:spanner-alg}. 
\end{proof}

\section{Further Applications of Streaming Algorithms for Connectivity Augmentation}
\label{sec:further}
In this section, we show applications of our streaming algorithms for $k$-CAP for following well-studied {\em network design} problems: STAP, SNDP and $k$-ECSS.

\subsection{Steiner Tree Augmentation Problem (STAP) in Streaming}
In STAP, we are given a set of vertices $V$ partitioned into {\em terminal} nodes $(R)$ and {\em Steiner} nodes $(V\setminus R)$, 
and a Steiner tree $T$ spanning the terminal set $R$. Then given a set of weighted links $L \subseteq \binom{V}{2}$, the goal is to find a minimum weight set of links $S \subseteq L$ such that $H = (V, E(T) \cup S)$ has $2$ edge-disjoint paths between any pair of terminals. The problem is a special case of SNDP and can be approximate within a factor of $2$ by iterative rounding method of~\citet{Jain01}. In light of recent developments for approximating tree augmentation and connectivity augmentation problems,~\citet{ravi2022new} provided a $(1.5+\eps)$-approximation for Steiner tree augmentation problem in polynomial time.

\paragraph{Algorithm in fully streaming setting.}
First, we observe that our results imply an algorithm for STAP in the fully streaming setting.

\begin{corollary}
\label{cor:steinerfully}
STAP in the fully streaming model can be solved by a single-pass streaming algorithm with approximation ratio $(2t-1+\eps)$ and space complexity $O(\eps^{-1}n^{1+1/t}\log n)$ words.
\end{corollary}
\begin{proof}
It basically follows from the same proof of Theorem~\ref{thm:fullystreamingkconn}: we build a $2$-connectivity certificate $E'\subseteq E$, and a $(2t-1+\eps)$-spanner $L'\subseteq L$, and solve the STAP on $(E',L')$.
Since $E'$ preserves the $2$-cuts between every pair of vertices, the same argument for feasibility and approximation extends to the Steiner setting as well.
\end{proof}
Note that the same fully streaming algorithm from Corollary~\ref{cor:steinerfully} can also be used to solve STAP in the easier link arrival stream.

\paragraph{Lower bound in link arrival streams.}
Now we show that STAP has a lower bound nearly matching Corollary~\ref{cor:steinerfully} \emph{even in link arrival streams}. This shows a separation of STAP from the easier TAP: the latter problem has a better streaming algorithm in link arrival streams than in the fully streaming setting, whereas the former problem does not.

\begin{corollary}
For any constant integer $t\ge 1$, weighted STAP in link arrival streams requires space complexity $\Omega(n^{1+1/t})$ bits (assuming the Erd\H{o}s's girth conjecture) to approximate the solution cost to a factor better than $2t+1$.
\label{cor:lowerbound-steiner}
\end{corollary}
\begin{proof}
   We use a similar (and simpler) construction as in the proof of  Theorem~\ref{thm:lowerbound-t}.
   
   Let $G$ be a fixed graph on $|V|=n$ vertices with girth $>2t+1$ and $|E(G)| = \Omega(n^{1+1/t})$ edges (which exists assuming the Erd\H{o}s's girth conjecture). Again, consider the INDEX  problem where Alice has a subgraph $H\subseteq G$. Alice sends a message to Bob. Then, Bob needs to decide whether $(u,v)\in H$ for a given edge $(u,v) \in G$.
 This task requires message size $\Omega(|E(G)|)$ bits. 

Now we use the streaming algorithm $\mathcal{A}$ (in the link arrival model) for STAP to design a protocol for the INDEX problem. 
Consider the STAP instance on a graph $G'=(V',E')$, where $V':= V \cup \{s,t\}$, and the terminal set is $R=\{s,t\}$, and the Steiner tree $E'$ has a single edge $(s,t)$.  Now, Alice and Bob together define the link set $L$ of this STAP instance: Alice adds to $L$ all her edges $E(H)$ as links with weight $1$. Then, Bob adds to $L$ two links $(s,u)$ and $(t,v)$ with weight $0$.  Then, the optimal solution for augmenting the Steiner tree $(s,t)$ is the length of the shortest path  from $s$ to $t$ using links from $L$. If $(u,v)\in E(H)\subset L$ then the shortest path length is $w(s,u)+w(u,v)+w(v,t)=0+1+0=1$. Otherwise, since $H\subseteq G$ has girth $>2t+1$, the shortest path length is at least $2t+1$.
Hence, any better-than-$(2t+1)$ approximation to STAP can solve this INDEX problem and thus needs space $\Omega(n^{1+1/t})$ bits.

For the unweighted STAP, the argument above yields the same space lower bound for better-than-$\frac{2t+3}{3}$ approximation.
\end{proof}

\subsection{SNDP in Edge Arrival Streams}\label{sec:sndp}
In this section, using our results and techniques from $k$-CAP and weighted spanners, we present a streaming algorithm for the general SNDP problem in edge arrival streams. We remark that our result in this section provide coresets for {\em covering functions} defined on cuts. 
\begin{lemma}
\label{lem:coreset}
Consider a weighted graph $G=(V,E)$ in an edge arrival stream. For integer $k\ge 1$ there is a one-pass streaming algorithm that computes $k$ disjoint edge subsets $S_1\uplus S_2\uplus \dots \uplus S_k \subseteq E$ each of size $|S_i|\le O(\eps^{-1}n^{1+1/t}\log n)$, in total space $O(k\eps^{-1}n^{1+1/t}\log n)$ words such that, for every $i\in [k]$ and every $e=(u,v)\in E\setminus (S_1\cup S_2\cup \dots \cup S_i)$, there is be a path $P\subseteq S_i$ connecting $u,v$  with total length $w(P) \le (2t-1+\eps)w(e)$.
\end{lemma}
\begin{proof}
We run $k$ instances $\mathcal{A}_1,\dots,\mathcal{A}_k$ of our streaming algorithm for spanner (Theorem~\ref{thm:spanner-alg}) in parallel. When an edge $e\in E$ arrives from the stream, we first feed it to $\mathcal{A}_1$.
This step may cause $\mathcal{A}_1$ to remove some edges from its memory (it could be that $e$ itself is not stored, or some other stored edges are evicted). 
We feed all these edges removed  by $\mathcal{A}_1$ into $\mathcal{A}_2$, and repeat the similar process, and so on.
Formally, let procedure $\textsc{INSERT}(\mathcal{A}_i,e)$ inserts $e$ to algorithm $\mathcal{A}_i$, and returns the set of edges evicted by $\mathcal{A}_i$. Starting with $F_0=\{e\}$, we iterate over $i=1,2,\dots,k$ and let $F_{i}$ be the union of the return values of $\textsc{INSERT}(\mathcal{A}_i,e')$ over all $e'\in F_{i-1}$. 

Finally return the edges stored by $\mathcal{A}_1,\dots,\mathcal{A}_k$ as $S_1,\dots,S_k$. Observe that they are disjoint subsets of $E$ by construction, and the size bound and space bound follow from Lemma~\ref{lem:spannerspacebound}.
By construction, if $e\in E \setminus (S_1\cup \dots \cup S_{i-1})$, then $e$ must have been fed into $\mathcal{A}_i$ at some point.
By the property of the spanner algorithm (see proof of Lemma~\ref{lem:spannerstretch}), if an edge $e$ is ever fed into $\mathcal{A}_i$ but eventually not stored in $S_i$, then there is a path $P\subseteq S_i$ that approximates $e$ as claimed.  This finishes the proof of the desired property of $S_1,\dots,S_k$.
\end{proof}

One of the main algorithmic approaches for SNDP is the augmentation framework pioneered by~\citet{williamson1993primal}. In this approach, the solution is constructed in $k$ phases and by the end of the phase $\ell$, the connectivity of every pair $u,v$ in the so-far-constructed solution is at least $\min\{\ell, r(st)\}$. So, the optimization problem of each phase is to increase connectivity of subset of pairs by one. More precisely, in each phase $\ell$, we need to pick a minimum-weight subgraph $H$ to cover a function $f_{\ell}: 2^{V} \rightarrow \{0,1\}$. We say that a subgraph $H$ covers $f$ iff for every $U\subset V$, $\delta_H(U) \ge f(s)$. In the case of SNDP, for every $\ell\le k$, $f_\ell$ is a skew-supermodular function and admits a $2$-approximation via a primal-dual algorithm~\citep{williamson1993primal}.   

Next, We use Lemma~\ref{lem:coreset} to show a coreset for covering $\{0,1\}$ functions $f: 2^V \rightarrow \{0,1\}$:
\begin{definition}
\label{defn:generalproblem}
Given a weighted graph $G = (V, E)$, and a function $f: 2^V \to \{0,1,\dots,k\}$, find an edge subset $H\subseteq E$ with minimum total weight such that for all $U\subseteq V$ it holds that $|\delta_H(U)| \ge f(U)$. Throughout this section, we consider the functions $f$ arising from an instance of SNDP on $G$ with connectivity requirement function $r$ with maximum requirement $k$. Then, for every $U\subset V$, $f(U) := \max_{s\in U, t\in V\setminus U} r(st)$.\footnote{All results hold for a more general class of {\em proper} functions too. The function $f$ is proper if $f(V) =0$, $f(U) = f(V\setminus U)$ for every $U\subset V$ (symmetry), and $f(U_1 \cup U_2)\le \max\{f(U_1), f(U_2)\}$ whenever $U_1$ and $U_2$ are disjoint (maximality).}
\end{definition}
\begin{lemma}
\label{lem:corsetapprox}
Given a weighted graph $G=(V,E)$, let $S=S_1\cup \dots \cup S_k$ be the set of edges returned by the algorithm of Lemma~\ref{lem:coreset}. 
Then, the optimal solution for covering a function $f:2^V \rightarrow \{0,1,\cdots, k\}$ (arising from a SNDP instance on $G$) on graph $G' = (V,S)$ is an $O(t\log k)$-approximation of the optimal solution for covering $f$ on $G=(V,E)$.
\end{lemma}
\begin{proof}
We use the augmentation framework to show the existence of an $O(t\log k)$-approximation coreset $S$. Note that while the augmentation approach of~\citep{williamson1993primal} achieves a $(2k)$-approximation,~\citet{goemans1994improved} showed by doing the augmentation in the reverse order, the approximation guarantee improves to $2\log k$. More precisely, in the phase $\ell$ of the reverse order augmentation, the optimization problem is to cover $\overleftarrow{f}_\ell$ where for every $U \subset V$, $\overleftarrow{f}_\ell(U) =  \max(0, \max_{s\in U, t\in V\setminus U} r(st) - (k-\ell))$. In words, by the end of the first phase, all pairs with connectivity requirement becomes connected and by the end of phase $\ell$, all pairs with connectivity requirements at least $\ell + m$, where $m\ge 1$, has at least $m$ edge-disjoint paths in the so-far-constructed solution.~\cite{goemans1994improved} proved that if there exits an $\alpha$-approximation for the optimization problem arises in the reverse order augmentation, the overall algorithm returns an $O(\alpha\cdot \log k)$-approximation for covering $f$ arises from an instance of SNDP (or more generally, proper functions).

Now we show that $G'$ is a $O(t\log k)$-approximate coreset for covering $f$. Start with empty graph $H_0=(V,\emptyset)$. For each $i=1,2,\cdots, k$, consider the set of links 
$L_i:= S_1\cup \cdots \cup S_{i}\setminus H_{i-1}$, and let $F_i\subseteq L_i$ be the minimum weight subset of links such that the augmented solution $H_{i-1} \cup F_i$ satisfies $\delta_{H_{i-1}\cup F_i}(U)\ge \overleftarrow{f}_i(U)$ for all $U\subseteq V$.
Let $H_i:= H_{i-1} \cup F_i$. Finally, output $H_k$. Note the invariant $H_i\subseteq S_1\cup \dots \cup S_i$.
We first show that for each phase $i$, $F_i\subseteq L_i$ is a ``good'' augmentation set compared to the best augmentation set in $E\setminus H_{i-1}$.
\begin{claim}
For $i\le k$, consider the minimum weight $\opt_i\subseteq E\setminus H_{i-1}$ such that $\delta_{H_{i-1}\cup \opt_i}(U)\ge \overleftarrow{f}_i(U)$ for all $U\subseteq V$.
Then, $w(F_i) \le (2t-1+\eps) \cdot w(\opt_i)$.
\end{claim}
\begin{proof}
By definition of our algorithm, the current solution $H_{i-1}$ should already satisfies $|\delta_{H_{i-1}}(U)| \ge \overleftarrow{f}_{i-1}(U)$ for all $U\subseteq V$.

It suffices to construct a small-weight solution $\tilde F \subseteq L_i = S_1\cup \dots \cup S_i \setminus H_{i-1}$ such that $|\delta_{H_{i-1}\cup \tilde F}(U)| \ge \overleftarrow{f}_i(U)$ for all $U\subseteq V$. Note that since $H_{i-1}\subseteq S_1\cup \dots \cup S_{i-1}$, $S_i \subseteq L_i$. For each $e=(u,v)\in \opt_i$, note that $e\in \opt_i \subseteq E\setminus H_{i-1} = \big (E\setminus (S_1\cup \dots \cup S_{i})\big )\cup L_i$. \begin{itemize}
     \item If $e\in L_i$, then add $e$ to  our solution $\tilde F$.
     \item Otherwise, $e\in E\setminus (S_1\cup \dots \cup S_{i})$. By Lemma~\ref{lem:coreset},  there is a path $P\subseteq S_i \subseteq L_i$ whose total weight is not more than $(2t-1+\eps) \cdot w(e)$. We add all edges of $P$ to $\tilde F$.
 \end{itemize}    
We have $\tilde F\subseteq L_i$ by construction (note that we do not keep duplicates in $\tilde F$). It is clear that $w(\tilde F) \le (2t-1+\eps)\cdot w(\opt_i)$. It remains to show $\tilde F$ is feasible:

Consider any $U\subseteq V$ that is previously not satisfied, i.e., $|\delta_{H_{i-1}}(U)| < \overleftarrow{f}_{i}(U)$. This means $\delta_{H_{i-1}}(U) = \overleftarrow{f}_{i-1}(U) = \overleftarrow{f}_{i}(U)-1$. By the feasibility of $\opt_i$, there is an edge $e\in \opt_i$ such that $e\in \delta_{H_{i-1}}(U)$. Then by our construction there is also an edge $e'\in \tilde F$ such that $e'\in \delta_{H_{i-1}}(U)$, which implies $|\delta_{H_{i-1}\cup \tilde F}(U)|\ge \overleftarrow{f}_{i}(U)$ as desired.
\end{proof}

Finally, as we show that the set $F_i$ is a $(2t-1+\eps)$-approximate solution of the optimization problem of phase $\ell$ in the reverse augmentation framework, by the result of~\citep{goemans1994improved}, the set $\bigcup_{i\in [k]} F_i$ is an $O(t \log k)$-approximate solution for covering the function $f$. Hence, $G'$ is an $O(t\cdot \log k)$-approximate coreset for covering $f$ and the proof is complete.
\end{proof}

\begin{theorem}\label{thm:sndp}
SNDP with maximum connectivity requirement $k$ on a weighted graph $G=(V,E)$ admits a single-pass streaming algorithm with space complexity $O(kn^{1+1/t})$ words and approximation ratio $O(t\log k)$.
\end{theorem}
\begin{proof}
    We first use Lemma~\ref{lem:coreset} to compute the edge set $S$, and solve the SNDP instance on the subgraph $(V,S)$ optimally by an exponential-time exact algorithm (with no space complexity overhead).  Then Lemma~\ref{lem:coreset} ensures that our solution achieves $O(t\log k)$ approximation. Note that the $\eps^{-1}$ and $\log n$ factor can be  omitted from the space complexity of Lemma~\ref{lem:coreset}, by setting $\eps = \Theta(1)$ and adjusting $t$ by a constant, which only affect the hidden constant factor in front of the approximation ratio $(2t-1+\eps)\cdot O(\log k)$.
\end{proof}
Note that SNDP generalizes STAP, so the same lower bound for STAP from Corollary~\ref{cor:lowerbound-steiner} also applies to SNDP. Specifically, for any constant integer $t\ge 1$, weighted SNDP requires space complexity $\Omega(n^{1+1/t})$ bits (assuming the Erd\H{o}s's girth conjecture) to approximate the solution cost to a factor better than $2t+1$.

\subsection{Min-Weight \texorpdfstring{$k$}{k}-ECSS in Edge Arrival Streams} 
As a corollary of Theorem~\ref{thm:link-arrival-alg} for connectivity augmentation in link arrival stream, we have the following guarantee for the problem of finding minimum-weight $k$-edge-connected spanning subgraph ($k$-ECSS). Formally, in weighted $k$-ECSS, given a graph $G=(V, E)$ with a weight function $w: E\rightarrow \mathbb{R}_{\ge 0}$, the goal is to find a minimum-weight $k$-edge-connected subgraph $H\subseteq G$.  

\begin{corollary}\label{thm:k-ecss-alg}
There exists a $k$-pass $O(\log k)$-approximation algorithm for weighted {$k$-ECSS} with total memory space $O(nk + n \log \min(n,W))$ where $W= \max_{e\in E} w(e)$.
\end{corollary}
\begin{proof}
We start with an empty graph $H_0$ and make $k$ passes over the stream of edges in $G=(V,E)$ where in each pass the goal is to augment the connectivity by one.  
More precisely, in the $\ell$th pass, we augment the edge-connectivity of $H_{\ell-1} = (V, E_{\ell-1})$ by one and construct $H_\ell = (v, E_\ell)$. Note that in each pass, the set of edges in available for augmentation is $L_\ell := E\setminus E_{\ell-1}$. The goal of $\ell$th pass is to solve $\ell$-connectivity augmentation problem on $(H_{\ell-1}, L_\ell)$ in a link arrival stream.

To analyze the space complexity note that by the end of the $\ell$th pass, the constructed graph $H_\ell = (v, E_\ell)$ is $\ell$-edge-connected and has $O(\ell\cdot n)$ edges. This follows since $H_{\ell-1}$ is a minimal $(\ell-1)$-edge-connected graph. 

Next, we bound the approximation factor. While $O(k)$-approximation is a trivial bound, via a more careful analysis we can show $O(\log k)$-approximation. 

Let $H^* = (V, E^*) \subseteq G$ be a minimum weight $k$-edge-connected subgraph. Consider the standard LP-relaxation of connectivity augmentation problem. 
\begin{align}
&\auglp(H = (V, E), L)\nonumber\\[1mm]
\text{minimize }& \rlap{$\sum_{e\in L} w(e) \cdot x(e)$} \nonumber\\[1mm]
\text{s.t.}\qquad &\sum_{e\in \delta_{L}(S)} x(e) \geq \ell - |\delta_{E}(S)| &&\forall S\subset V \label{cst:cover} \\ 
&0 \le x(e) \le 1  &&\forall e\in L
\end{align}
First we show that there exists a fractional solution $x_{\ell}$ for $\auglp(H_{\ell-1}, L_i)$ of cost at most $\frac{1}{k - \ell-1} \cdot w(H^*)$. Then, by the fact that the integrality gap of $\auglp$ is at most $2$ (e.g.,~\citep{Jain01}), an optimal integral solution for the connectivity augmentation instance in the $\ell$th pass has weight $O(\frac{1}{k - \ell-1} \cdot w(H^*))$. Hence, by the guarantee of our algorithm for connectivity augmentation (i.e., Theorem~\ref{thm:link-arrival-alg}), the total weight of the final solution of the algorithm at the end of $k$th pass, $H_k$, is
\begin{align*}
    w(H_k) = \sum_{\ell=1}^k O(\frac{1}{k - \ell-1} \cdot w(H^*)) = O(\log k \cdot w(H^*)).
\end{align*}
Now, we describe the fractional solution $x_\ell$. For every edge $e\in L_{\ell}$, if $e\in E^*\setminus E_{\ell-1}$ then $x_\ell(e) = \frac{1}{k-\ell+1}$; otherwise, $x_\ell(e) =0$. It is straightforward to check that the cost of the fractional solution $x_\ell$ is at most $\frac{1}{k-\ell+1} \cdot w(H^*)$. 
Next, we show that $x_\ell$ is a feasible solution for $\auglp$ on $(H_{\ell-1}, L_{\ell})$. Consider a subset $S\subset V$. If $|\delta_{E_{\ell-1}}(S)| \ge \ell$ then constraint~\eqref{cst:cover} trivially holds. Otherwise, since $|\delta_{E^*}(S)| \ge k$, $|\delta_{E^*\setminus E_{\ell-1}}(S)| \ge k-\ell+1$. So, 
\begin{align*}
    \sum_{e\in \delta_{L_i}(S)} x_\ell(e) \ge \sum_{e\in \delta_{E^*\setminus E_{\ell-1}}(S)} x_\ell(e) \ge 1 = \ell - |\delta_{E_{\ell-1}}(S)|.
\end{align*}
Hence, $x_\ell$ is a feasible solution for $\auglp(H_{\ell-1}, L_{\ell})$
\end{proof}

\begin{remark}
    The $O(\log k)$-approximation analysis in the proof of Corollary~\ref{thm:k-ecss-alg} is similar to the analysis of~\citep{goemans1994improved} for the augmentation framework of Survivable Network Design Problem in reverse order. Although, the reverse order augmentation problem in general is different from the standard (forward) augmentation framework, in the case of $k$-ECSS these two approaches solve the same ``augmentation'' problems and the analysis of~\citep{goemans1994improved} for reverse order augmentation framework proves $O(\log k)$-approximation for $k$-ECSS via the standard augmentation framework. 
\end{remark}

\bibliographystyle{abbrvnat}
\bibliography{streaming-network-design}
\end{document}